\documentclass[dvipsnames, 11pt]{article}

\usepackage{geometry}
\usepackage[utf8]{inputenc}
\usepackage[english]{babel}
\usepackage{lmodern}
\usepackage{microtype}

\usepackage{amsmath,amsfonts,amssymb}
\usepackage{amsthm}

\usepackage{graphicx}
\usepackage{authblk}
\usepackage{standalone}

\usepackage{hyperref}
\usepackage{xcolor}
\usepackage{standalone}
\usepackage{adjustbox}
\usepackage[capitalise]{cleveref}
\usepackage{thm-restate}

\usepackage[shortlabels]{enumitem}

\usepackage{braket}
\usepackage{wrapfig}

\title{Complete Compositional Syntax for Finite Transducers on Finite and Bi-Infinite Words}

\author{Titouan Carette}
\affil{\href{mailto:titouan.carette@polytechnique.edu}{titouan.carette@polytechnique.edu}\\
  LIX, CNRS, École polytechnique, Institut Polytechnique de Paris, 91120
    Palaiseau, France}

\author{Marc de Visme}
\affil{\href{mailto:marc.de-visme@inria.fr}{marc.de-visme@inria.fr}\\
  Université Paris-Saclay, CNRS, ENS Paris-Saclay, Inria, Laboratoire Méthodes Formelles, 91190 Gif-sur-Yvette, France}

\author{Vivien Ducros}
\affil{\href{vivien.ducros@ens-paris-saclay.fr}{vivien.ducros@ens-paris-saclay.fr}\\
  ENS Paris-Saclay, Université Paris-Saclay, 91190 Gif-sur-Yvette, France}

\author{Victor Lutfalla\footnote{Partially supported by ANR ALARICE (ANR-24-CE48-7504)}}
\affil{\href{victor.lutfalla@math.cnrs.fr}{victor.lutfalla@math.cnrs.fr}\\
  Aix-Marseille Univ, CNRS, I2M, Marseille, France}

\author{Etienne Moutot\footnote{Partially supported by ANR IZES (ANR-22-CE40-0011)}}
\affil{\href{mailto:etienne.moutot@math.cnrs.fr}{etienne.moutot@math.cnrs.fr}\\
CNRS, I2M, Aix-Marseille Université, Marseille, France}

\usepackage{tikzit}
\usepackage{stmaryrd}
\usepackage{etoolbox}
\usepackage{wrapfig}
\usepackage{dsfont} 

\tikzstyle{black dot}=[inner sep=0.4mm,minimum width=1.5mm,minimum height=0pt,fill=black,draw=black,shape=circle]

\tikzstyle{ltriangle-red}=[shape=isosceles triangle, line width=1pt, tikzit fill={rgb,255: red,180; green,180; blue,180}, tikzit draw={rgb,255: red,255; green,0; blue,0}, draw={rgb,255: red,155; green,0; blue,0}, isosceles triangle stretches=true, inner sep=0.8pt, minimum width=0.4cm, minimum height=0.4cm, shape border rotate=180]
\tikzstyle{rtriangle-red}=[shape=isosceles triangle, line width=1pt, tikzit fill={rgb,255: red,75; green,75; blue,75}, tikzit draw={rgb,255: red,255; green,0; blue,0}, draw={rgb,255: red,155; green,0; blue,0}, isosceles triangle stretches=true, inner sep=0.8pt, minimum width=0.4cm, minimum height=0.4cm]
\tikzstyle{lsignal-blue}=[shape=signal, signal to=west, signal from=east, tikzit shape=rectangle, line width=1pt, tikzit fill={rgb,255: red,180; green,180; blue,180}, tikzit draw={rgb,255: red,0; green,0; blue,255}, draw={rgb,255: red,0; green,0; blue,155}, minimum height=6pt, inner sep=1pt, font={\scriptsize}, tikzit category=GLA]
\tikzstyle{rsignal-blue}=[shape=signal, signal to=east, signal from=west, tikzit shape=rectangle, line width=1pt, tikzit fill={rgb,255: red,75; green,75; blue,75}, tikzit draw={rgb,255: red,0; green,0; blue,255}, draw={rgb,255: red,0; green,0; blue,155}, minimum height=6pt, inner sep=1pt, font={\scriptsize}, tikzit category=GLA]

\tikzstyle{ltriangle-black}=[shape=isosceles triangle, tikzit fill={rgb,255: red,180; green,180; blue,180}, draw=black, isosceles triangle stretches=true, inner sep=0.8pt, minimum width=0.4cm, minimum height=0.4cm, shape border rotate=180]
\tikzstyle{rtriangle-black}=[shape=isosceles triangle, tikzit fill={rgb,255: red,75; green,75; blue,75}, draw=black, isosceles triangle stretches=true, inner sep=0.8pt, minimum width=0.4cm, minimum height=0.4cm]
\tikzstyle{lsignal-black}=[shape=signal, signal to=west, signal from=east, tikzit shape=rectangle, tikzit fill={rgb,255: red,180; green,180; blue,180}, draw=black, minimum height=6pt, inner sep=1pt, font={\scriptsize}, tikzit category=GLA]
\tikzstyle{rsignal-black}=[shape=signal, signal to=east, signal from=west, tikzit shape=rectangle, tikzit fill={rgb,255: red,75; green,75; blue,75}, draw=black, minimum height=6pt, inner sep=1pt, font={\scriptsize}, tikzit category=GLA]

\tikzstyle{ppp}=[-, dashed, draw={rgb,255: red,140; green,140; blue,140}]
\tikzstyle{ddd}=[-, dotted, tikzit draw={rgb,255: red,128; green,0; blue,128}]
\tikzstyle{grey}=[-, tikzit draw={rgb,255: red,140; green,140; blue,140}, draw={rgb,255: red,140; green,140; blue,140}]
\tikzstyle{dblue}=[-, line width=1pt, tikzit draw={rgb,255: red,0; green,0; blue,255}, draw={rgb,255: red,0; green,0; blue,155}]
\tikzstyle{dred}=[-, line width=1pt,tikzit draw={rgb,255: red,255; green,0; blue,0}, draw={rgb,255: red,155; green,0; blue,0}]
\newcommand{\bfup}[1]{\textbf{\textup{#1}}}
\newcommand{\ZZ}{{\mathbb{Z}}}
\newcommand{\NN}{{\mathbb{N}}}

\newcommand{\blue}[1]{{\textcolor{blue!50!black}{#1}}}
\newcommand{\red}[1]{{\textcolor{red!50!black}{#1}}}
\newcommand{\binterp}[1]{\blue{\left\llbracket \textcolor{black}{#1} \right\rrbracket}}
\newcommand{\rinterp}[1]{\red{\left\llbracket \textcolor{black}{#1} \right\rrbracket}}
\newcommand{\D}{{\mathcal{D}}}
\renewcommand{\L}{{\mathcal{L}}}
\newcommand{\M}{{\mathcal{M}}}

\renewcommand{\P}{{\mathcal{P}}}
\newcommand{\Pne}{\mathcal{P}_{\neq \varnothing}}
\newcommand{\R}{{\mathcal{R}}}
\newcommand{\bR}{{\blue{\bfup{R}}}}
\newcommand{\rR}{{\red{\bfup{R}}}}
\renewcommand{\S}{{\mathcal{S}}}
\newcommand{\bS}{{\blue{\bfup{S}}}}
\newcommand{\rS}{{\red{\bfup{S}}}}
\newcommand{\T}{{\mathcal{T}}}
\newcommand{\bT}{{\blue{\bfup{T}}}}
\newcommand{\rT}{{\red{\bfup{T}}}}
\newcommand{\U}{{\mathcal{U}}}
\newcommand{\bU}{{\blue{\bfup{U}}}}
\newcommand{\rU}{{\red{\bfup{U}}}}

\newcommand{\bV}{{\blue{\bfup{V}}}}
\newcommand{\rV}{{\red{\bfup{V}}}}
\newcommand{\one}{{\ensuremath{\mathds{1}}}}
\newcommand{\shift}[2]{\triangleright^{%
		\ifstrempty{#1}%
		{\ifstrempty{#2}{}{#2}}%
		{\ifstrempty{#2}{#1}{#1\mid{}#2}}}}
\newcommand{\id}{\bfup{id}}
\newcommand{\transpose}[1]{#1^{t}}

\newcommand{\fact}[1]{\mathrm{Fact}(#1)} 
\newcommand{\limlang}[1]{\overrightarrow{#1}} 
\newcommand{\prun}[1]{\mathrm{Prun}(#1)} 
\newcommand{\closure}[1]{\mathrm{cl}(#1)} 
\newcommand{\orbit}[1]{\mathcal{O}(#1)} 
\newcommand{\internal}[1]{\sqsubseteq_{#1}} 
\newcommand{\regularlanguages}{\mathrm{Reg}}
\newcommand{\soficshifts}{\mathrm{Sofic}}
\newcommand{\regularfactorpruned}{\regularlanguages^\mathrm{fc}_\mathrm{p}}

\theoremstyle{plain}
\newtheorem{theorem}{Theorem}[section]
\newtheorem{lemma}[theorem]{Lemma}
\newtheorem{corollary}[theorem]{Corollary}

\newtheorem{proposition}[theorem]{Proposition}
\crefname{proposition}{Proposition}{Propositions}

\theoremstyle{definition}
\newtheorem{definition}{Definition}[section]
\theoremstyle{remark}
\newtheorem{remark}[theorem]{Remark}

\usepackage[
backend=biber,
style=alphabetic,
]{biblatex}
\begin{document}

\maketitle

\begin{abstract}
	Minimizing finite automata, proving trace equivalence of labelled transition systems or representing sofic subshifts involve very similar arguments, which suggests the possibility of a unified formalism. We propose finite states non-deterministic transducer as a lingua franca for automata theory, transition systems, and sofic subshifts. We introduce a compositional diagrammatical syntax for transducers in form of string diagrams interpreted as relations. This syntax comes with sound rewriting rules allowing diagrammatical reasoning. Our main result is the completeness of our equational theory, ensuring that language-equivalence, trace-equivalence, or subshift equivalence can always be proved using our rewriting rules.
\end{abstract}

\section{Introduction}
\label{sec:introduction}
The interest of category theory for computer science is to provide an abstract language suitable to describe computational structures in their most general form, allowing to clarify connections between similar constructions appearing in different fields, sometimes seemingly unrelated. 
One of those abstract ideas, ubiquitous in computer science, is a machine interacting with its environment by updating an internal set of states. 
Most computational models fit this description, the paradigmatic example being Turing machines. 
Such machines interact with their environment via inputs and outputs, thus their behavior can be described as the set of all possible input/output pairs, which are interpreted as a relation between a (possibly infinite) sequence of input events and another sequence of output events. 
When constraining the set of internal states of a machine to be finite, one enforces interesting mathematical properties on the set of admissible behaviors, as they must be definable by finite means, in a certain way.

Let us take three examples of models exhibiting this kind of behavior. 
The most famous is the finite automaton \cite{DBLP:books/daglib/0086373,DBLP:books/daglib/0016921}, which only takes finite sequences of input events and evaluates if those sequences are valid or not.
The behavior then coincides with the language recognized by the automaton, and the possible behaviors are exactly the regular languages.
Perhaps more aligned with the interacting machine point of view, a Labelled Transition System (LTS) performs a transition (depending on an input event and its current state) that can produce an output event. In the LTS literature, we usually talk about the trace of the system to design its behavior.
The last model we will mention is perhaps less known by the computer science community and comes from dynamical systems.
A discrete-time dynamical system is a set $X$ equipped with an update function $f:X\to X$ implementing the dynamic. 
Symbolic dynamics is the study of coarse-grained dynamical systems when $X$ is partitioned into subsets indexed by finitely many symbols \cite{lind-marcus}. 
One can learn a lot on the original dynamical system by studying the sequences of symbols induced by the evolution $f$.
Those sequences correspond to our general notion of behavior, and the ones "definable by finite means" correspond to a set of sequences known as sofic subshifts.
The fact that similar mathematical methods, mostly involving regular languages, can be employed to handle those three cases is folklore. 
However, to our knowledge, no framework exists to formally unify those three models.

We propose finite relational transducers as general objects subsuming automata, LTS, and symbolic dynamical systems. 
The originality of our approach is to embrace a fully relational point of view, considering relations and not functions as more fundamental. 
In other words, we will consider non-determinism to be more natural, and see determinism as an interesting special case. 
This stance will lead us to present generalizations of well-known models that might feel unfamiliar, typically in the case of non-deterministic symbolic dynamical systems. 
Still, we believe that the relational setting is the right place to understand clearly how the various models we present are linked.
Furthermore, the relational point of view is also the most well-suited for the use of string diagrams.

Indeed, our approach fits in a recent thread of research aiming to represent and reason on computational processes using diagrams with inputs and outputs.
Those diagrams can be thought of as boolean circuits-like structures, built from elementary generators (or gates) on which we can perform local rewritings by replacing a sub-circuit with another one having the same type and an equivalent behavior. Such methods have been successfully applied to provide compositional syntax for control flow graphs \cite{bonchi2014categorical,bonchi2021survey} or quantum computing \cite{van2020zx} for example.
As with the choice of moving from functions to relations, the diagrammatic paradigm requires the introduction of unfamiliar notions (for a reader used to automata, transition systems, or symbolic dynamic) from category theory \cite{selinger2011survey}.
Still, even if dressed in the unusual language of category theory and diagrams, the fundamental notions involved are the same.

We hope that the benefits of this choice outweigh the price paid for it.
Indeed, our formalism being formulated in the language of string diagrams (or equivalently, in the language of symmetric monoidal categories \cite{MacLane}), we expect to be able to obtain straightforward generalizations by changing our base category (the mathematical world in which we interpret our diagrams) from relation to stochastic kernels or quantum channels. 
This provides solid theoretical bases to define meaningful notions of automata, transition systems, and symbolic dynamical systems in the probabilistic and quantum case, on which we could develop similar simulation-based proof techniques.

One can see the current paper as the first step toward this goal. The use of diagrammatic formalism to unify notions of finite tilings and quantum tensors has been developed in \cite{quantumWang}.

We hope our unified formalism can ease the transfer of ideas between the transition system and symbolic dynamics communities and stimulate the emergence of new questions and insights set off by our diagrammatical approach. 

\smallskip

In this paper, we introduce relational transducers with finite states and show how they generalize finite automata, LTS, and symbolic dynamical systems. We define their behaviors on both finite and bi-infinite words. 
In both cases, we fully characterize the admissible behaviors with the notion of regular and sofic relations, respectively. 
We also introduce two string diagrammatic equational theories which are able to represent transducers on finite and bi-infinite words in a compositional way. 
We show that those equational theories are expressive enough to represent any transducers and present two sound simulation principles that are strong enough to equate in a diagrammatic form any two transducers having the same behavior, both in the finite and bi-infinite case. 
On a more technical side, we present a new diagrammatical formulation of the notion of backward and forward simulation and introduce a new notion of normal form for the presentation of a sofic subshift.

\smallskip

We start by introducing the diagrammatic aspects of our approach in \Cref{sec:relations}. 
In \Cref{sec:finite}, with the notion of transducers used throughout the paper, we introduce our diagrammatical equational theory and completeness result for transducers acting on finite words. 
\Cref{sec:infinite} gathers results on sofic relations culminating in the extension of our diagrammatical theory and a completeness proof for the case of bi-infinite words. 
Most of the proofs are postponed to the appendix.

\section{Relations and their Diagrams}
\label{sec:relations}

\label{sec:logic}

Before introducing transducers, we quickly present the formalism we use to handle relations between finite sets. We later use the same formalism for relations between infinite sets, more specifically sets of finite words and sets of bi-infinite words.
In the context of relations, those infinite sets will always be clearly denoted with a $\_^*$ or a $\_^\ZZ$, as such, the letters $A,B,C,\dots$ will refer to finite sets.

Given some finite sets $A_1,\dots,A_n,B_1,\dots,B_m$, a relation $\R$ from $\Pi_{i=1}^n A_i$ to $\Pi_{j=1}^m B_j$ is simply a subset $\R \subseteq \Pi_{i=1}^n A_i \times \Pi_{j=1}^m B_j$. We write $a~\R~b$, and say that $\R$ relates $a$ to $b$, whenever $(a,b) \in \R$. We represent such relation as the diagram on the left, the identity relation on $\Pi_{i=1}^n A_i$ as the one on the middle, and the transposed of $\R$ (defined by $b\R^t a$ iff $a\R b$) as the one on the right.
\[ \input{rel-def.tikz} \qquad \qquad \input{rel-identity.tikz} \qquad \qquad \input{rel-transposed.tikz} \]

We write $\one = \{()\}$ for the singleton set which is the zero-ary Cartesian product, and we will often omit diagrammatically the wires labeled by $\one$. Three notable relations are the swap $\gamma_{A,B} = \{ ((a,b),(b,a)) \mid a \in A, b \in B\}$, the cap $\eta_A = \{ ((),(a,a)) \mid a \in A\} : \one \to A \times A$ and the cup $\epsilon_A = \{((a,a),()) \mid a \in A\} : A \times A \to \one$, which we represent with bent wires as below. We also represent below the ``full'' relation $\bullet_{A,B} = A \times B : A \to B$.
\[ \input{rel-swap-cup-cap.tikz} \]
\begin{figure}[!h]
	\input{rel-equations-cup-cap.tikz}
	\caption{Equations for the Cup and Cap.} 
	\label{fig:rel-equations-cup-cap}
\end{figure}
We give in \Cref{fig:rel-equations-cup-cap} a couple of remarkable identities about them.
We note that there is a lot of arbitrary choices in those representations, as for example a relation from $A \times B$ to $\one$ could be represented in any of the following ways.
\[ \input{rel-equivalent-rep.tikz}\]

Nevertheless, those graphical representations remain practical, especially when representing various kind of compositions. 
The usual ones are the sequential compositions of $\R : A \to B$ with $\S : B \to C$, and the parallel composition of $\R : A \to B$ and $\S : C \to D$, which are defined as follows:
\[ \input{rel-normal-composition.tikz}\]

\begin{wrapfigure}{r}{0.2\textwidth}
	\input{rel-partial-composition.tikz}
\end{wrapfigure}
The general case is the partial composition. For example, given two relations $\R : A \to B \times C$ and $\S : C \times D \to E$, their composite relation $\{ ((a,d),(b,e)) \mid \exists c, a~\R~(b,c), (c,d)~\S~e \} : A \times D \to B \times E$ is written diagrammatically as shown on the right. Conversely, whenever we have a diagram that decomposes $\R : \Pi_{i=1}^n A_i \to \Pi_{j=1}^m B_j$, into the partial composition of $\R_1,\dots,\R_r$, and if we write $C_1, \dots, C_p$ for all the sets labeling the ``internal wires'' of the diagram, then we can build a logical formula equivalent to $(a_1,\dots,a_n)~\R~(b_1,\dots,b_m)$ of the following form:
\[ \exists c_1 \in C_1, \dots, \exists c_p \in C_p, (\_ \R_1 \_) \land \dots \land (\_ \R_r \_) \]
where the blanks $\_$ have to be filled with adequate tuples of $a_i$, $b_j$ and $c_k$. For example, the following diagram would correspond to the following logical formula:
\[a~\left(\input{rel-logic-example.tikz}\right)~b \iff \exists c_1 \in C_1, \exists c_2 \in C_2, ((a,c_1)~\R_1~(c_2,b)) \land (c_1~\R_2~c_2) \]
\begin{wrapfigure}{r}{0.2\textwidth}
	\input{rel-reorganise.tikz}
\end{wrapfigure}

We remark that if any of the $\R_\ell$ is the full relation $\bullet_{A,B}$, then it can safely be omitted from the logical formula as $(a~\bullet_{A,B}~b)$ is always true.

Combining the partial composition with the swap, cup and cap defined above, one can reorganize the inputs and outputs of a relation, turning for example a relation $\R : A \times B \times C \to D$ into a relation from $B \times A$ to $D \times C$ as on the right. 

All those diagrams can actually be formalized using string diagrams from category theory  \cite{MacLane}. Indeed, finite sets and relations form a strict symmetric monoidal category, called \bfup{FinRel}. We recall the definition of such a category in \Cref{app:relations}, though for the sake of this paper, the only required understanding is ``the diagrams works as intended, there is never a need to explicitly add bracketing, and all 'reasonable' ways of rewriting a diagram yield the same relation''. If we account for the cup and cap, relations even form a compact closed category\footnote{Where the dual is the transposition on relations, and the identity on objects.}, which exactly means that it satisfies the equations listed in \Cref{fig:rel-equations-cup-cap}. We redirect to \cite{selinger2011survey} for a survey of variations around monoidal categories and the corresponding diagrammatic notations.

\section{Transducers on Finite Words}
\label{sec:finite}

\subsection{Transducers and Regular Relations}

For any finite set $A$, called an alphabet, we denote $A^*$ the set of finite words over $A$, that is $A^* = \biguplus_{k \geq 0} A^k$. 
We write a word $w = w_1\dots w_k \in A^*$, $|w|$ for its size $k$ and $\epsilon$ for the empty word. 
The concatenation of two words $u$ and $v$, respectively of size $n$ and $m$, is denoted $uv$ and has size $n+m$. 
A subset of $A^*$ is called a language over the alphabet $A$. A \textbf{uniform relation} between two alphabets $A$ and $B$, denoted $\mathcal{R}:A^*\to B^*$, is a language over the product alphabet $A\times B$, equivalently, it's a relation between $A^* $ and $B^* $ only relating words of the same size, \textit{i.e.} $u~\R~v $ implies $|u|=|v|$.

The product of two uniform relations $A^* \to B^* $ and $C^* \to D^* $ is formally a relation $A^* \times C^* \to B^* \times D^* $ but to maintain the uniformity constraints, we choose to restrict it to a uniform relation $(A\times C)^* \to (B\times D)^* $. 
While of course $(A\times B)^* \neq A^* \times B^*$ as sets, we will identify the two as long as all relations involved are uniform. In other words, uniform relations are ignoring the elements of $A^* \times B^* $ that are not in $(A\times B)^*$. We keep the notation $\R \times \S$ for this uniform product, even if it is no longer a Cartesian product. Uniform relations are preserved by composition and (uniform) product, and include identities, as well as cups and caps. Similarly to \bfup{FinRel}, they form a compact closed strict symmetric monoidal\footnote{Note that the monoidal unit of \bfup{UniRel} is $\one^*$, that is the words over the singleton alphabet, which is isomorphic to $\mathbb{N}$.} category called \bfup{UniRel}, meaning we can use the exact same graphical representation for them. 

\begin{definition}[Transducer]
	A (non-deterministic finite) transducer is a tuple $(\T, A, B, Q, I, F)$ where $A$ is a finite set representing the \emph{input alphabet}, $B$ is a finite set representing the \emph{output alphabet}, $Q$ is a finite set representing the \emph{states}, $\T: A\times Q \to B\times Q$ the \emph{transition relation}, ${I} \subseteq Q$ the \emph{initial states} and ${F} \subseteq Q$ the \emph{final states}.
\end{definition}

Similarly to non-deterministic finite automata, we write $q \xrightarrow[b]{a} q'$ whenever $((a,q),(b,q'))$ is a transition of $\T$. A \textbf{run} of length $k$ within a transducer is a sequence
\[ q_0 \xrightarrow[b_1]{a_1} q_1 \xrightarrow[b_2]{a_2} \dots \xrightarrow[b_k]{a_k} q_k \]
such that $q_0 \in I$ and $q_k \in F$. When such a run exists, we say that the transducer can transform the word $a_1\dots a_k$ into $b_1\dots b_k$. The \emph{behavior} of the transducer, written $\L(\T, A, B, Q, I, F)$, is the uniform relation from $A^*$ to $B^*$ which relates $w$ to $v$ whenever the transducer can transform $w$ into $v$. 

A transducer with only inputs, said otherwise $B=\one$, is exactly a non-determinisitic finite automaton. Alternatively, every transducer from $A$ to $B$ can be seen as a non-deterministic finite automaton over $A \times B$. We recall that a \textbf{regular language}, is a language recognized by a finite automaton. 
A \textbf{regular relation} is a uniform relation $A^* \to B^* $ which is a regular language when seen as a subset of $(A\times B)^*$, which is the case if and only if it is the behavior of a transducer. 

\begin{proposition}\label{prop:RegRel-is-cat}
	Regular relations are preserved by composition and product. We call \bfup{RegRel} the subcategory of \bfup{UniRel} of regular relations.
\end{proposition}

The proof is based on standard construction of products and composition of transducers.


Any relation $\mathcal{R}:A\to B$ can be lifted to a regular\footnote{The corresponding transducer has only one state which is accepting and initial and its transition relation is directly $\mathcal{R}$ seen as having type $A\times \one \to B\times \one $.} relation $\mathcal{R}^* : A^* \to B^* $ defined letter by letter: $\mathcal{R}^* = \{ (u,v) \mid |u| = |v|, \forall 1 \leq i \leq |u|, (u_i,v_i) \in \R\}$. Categorically, this lift is a faithful strong symmetric monoidal functor from \bfup{FinRel} to \bfup{RegRel}, which means:

\begin{description}
	\item[Functoriality] $(\S \circ \R)^* = \S^* \circ \R^*$ and $(\id_A)^* = \id_{A^*}$.
	\item[Strong Monoidality] $(\R \times \S)^*$ can be identified with $\R^* \times \S^*$ within \bfup{RegRel}.
	\item[Symmetric] $(\gamma_{A,B})^* = \gamma_{A^*,B^*}$.
	\item[Faithfullness] If $\R^* = \S^*$, then necessarily $\R = \S$.
\end{description}

This notion of transducers and regular relations appeared early in the history of theoretical computer science, one can find similar definitions, though in a different form in \cite{DBLP:journals/ibmrd/ElgotM65}.

\subsection{Diagrammatic Representation of Transducers}

Looking at $\R : \Pi_{i=1}^n A_i \to \Pi_{j=1}^m B_j$, we can see its lifted relation $\R^*$ as going from $\Pi_{i=1}^n A_i^*$ to $\Pi_{j=1}^m B_j^*$, since $(\Pi_{i=1}^n A_i)^*$ is identified with $\Pi_{i=1}^n A_i^*$ (we are only working with uniform relations here).
Diagrammatically, this allows us to apply $\_^*$ to individual wires:
\[ \input{rel-lifted.tikz}\]

When looking at uniform relations over words, one such regular\footnote{It is the behavior of the transducer $(\gamma_{A,A},A,A,A,I,F)$.} relation will be particularly useful, the \textbf{finite shift}  $\shift{I}{F}_A : A^* \to A^*$ for a finite set $A$ and two subsets $I,F \subseteq A$, defined as $\{ (w,v) \mid \exists i \in I, \exists f \in F, iw = vf \}$. Said otherwise, $w$ and $v$ are related by the shift if either both are empty and $I \cap F \neq \varnothing$, or if neither are empty and the first letter of $v$ is in $I$, the last letter of $w$ is in $F$, and the remaining letters satisfy $w_k = v_{k+1}$. Diagrammatically, we represent the shift and its transposed:
$\input{rel-shift-finite.tikz}$.

The shift allow to obtain the behavior of a transducer in a compositional way.

\begin{proposition}[Proved in \Cref{app:transducer-fin}]\label{prop:transducer-fin}
	For every transducer $(\T, A, B, Q, I,F)$, its behavior can be obtained by lifting $\T$ and using the shift as follows:
	\[ \input{rel-trans-fin.tikz} \]
\end{proposition}

\subsection{Diagrammatic Language for Transducers}\label{sec:diagrams-blue}

\begin{wrapfigure}{r}{0.2\textwidth}
		\input{graph-from-trans-fin.tikz}
\end{wrapfigure}
This is however not the only way to represent transducers using relations. Another approach is to start from the category \bfup{FinRel} and freely add a feedback, state that a transducer $(\T,Q,B,Q,I,F)$ should be able to be represented by the diagram on the right, and then find additional equations ensuring that this feedback operator behaves exactly as transducers do.

Formally, we define a graphical language \bfup{Trans}, so a syntactical construct, where the diagrams are seen as mathematical objects distinct from the relations they represent -- which is a category where the objects are finite sets, and where the morphisms are generated by composing sequentially and in parallel the generators below together with the usual equations of a strict symmetric monoidal category (see \Cref{appfig:rel-equations-fin} in the appendix for a reminder), the equations of a feedback  category\footnote{Some authors, like in \cite{Katis09}, have an additional equation which allows to slide isomorphisms through the feedback. While we do not have this equation, it is actually a special case of \Cref{fig:graph-slide}, which is part of the completed equational theory.} (see \Cref{fig:graph-feedback}), and some equations ensuring we faithfully embed \bfup{FinRel} (see \Cref{fig:graph-finrel}).
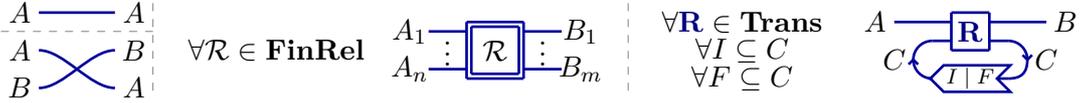
	\[\input{graph-generators.tikz}\]

\begin{figure*}
	\input{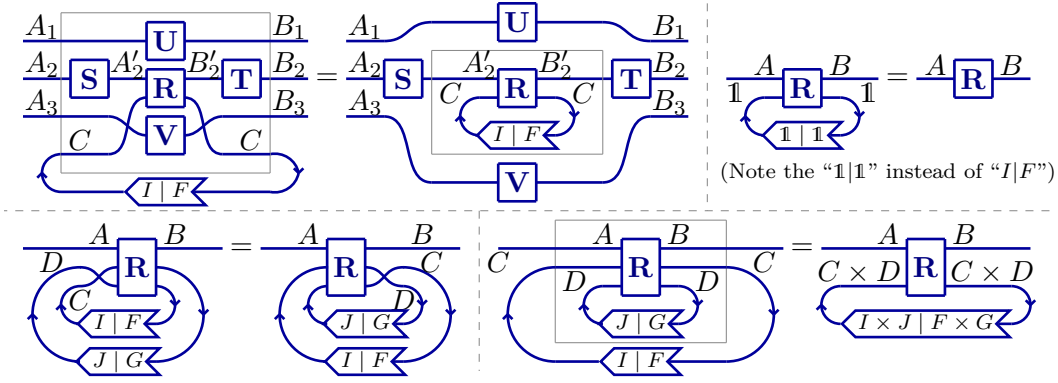}
	\caption{Equations for a Feedback Category.}
	\label{fig:graph-feedback}
\end{figure*}

\begin{figure*}
	\input{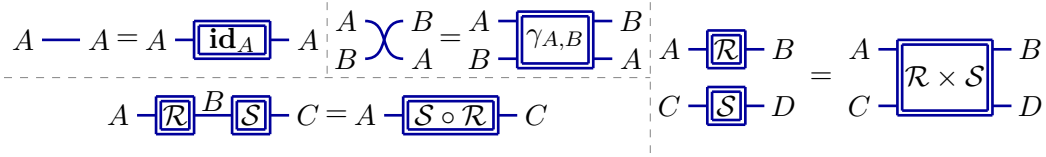}
	\caption{Equations for Faithfully Embedding \bfup{FinRel}.}
	\label{fig:graph-finrel}
\end{figure*}

Let us point a couple of facts about this language:
\begin{itemize}
	\item A bundle of wires labeled $A_1,\dots,A_n$ is the same as a single wire labeled $\Pi_{i=1}^n A_i$.
	\item The wires and syntactical constructs are in \bfup{\blue{thick blue}}, to distinguish them from actual relations. $\R$ refers to an actual relation while $\bR$ refers to an element of our language.
	\item Double-line boxes denote the generator, while single-line boxes denote any diagram potentially constituted of many generators, including feedbacks.
	\item The arrows on the feedback are a reminders that those are not the same as the cup and cap of \Cref{sec:relations}.
	\item Within \Cref{fig:graph-feedback}, the gray lines correspond to bracketing, and similarly to the equations of a strict symmetric monoidal category, the overall consequences of those equations is that bracketing can be safely ignored. Additionally, compared to the literature, those equations had to be adapted to the presence of labels on the feedback.
	\item The equations of \Cref{fig:graph-finrel} ensures that if a diagram does not contain any instance of the ``feedback'' generator, we can merge all the generators into a single double-line box.
	\item The current equational theory is incomplete, additional equations will be added in \Cref{fig:fin-simulation-principle} to obtain completeness.
\end{itemize}

In order to ensure that the equations are not contradictory\footnote{Which would lead to the trivial category where all the diagrams are equal to one another.}, we provide a semantics and prove soundness of our equations. The semantics is a strong symmetric monoidal functor from \bfup{Trans} to \bfup{UniRel}, which we write $\binterp{-}$, and is actually simply ``removing the color and adding a $\_^*$ everywhere''. We provide an explicit definition in the appendix (see \Cref{appfig:fin-interp}). Soundness means that if one rewrites a diagram $\bR$ into $\bS$ using any of the listed equations, we still have $\binterp{\bR} = \binterp{\bS}$.
Completeness would be the other way around, whenever $\binterp{\bR} = \binterp{\bS}$ then we could rewrite $\bR$ into $\bS$ (in other words,  $\binterp{-}$ would be faithful).
While not having completeness yet, we still have a quasi-normal form, and universality for regular relations.
		
\begin{proposition}[Quasi-Normal Form]\label{prop:normal-form-fin}~\\
	\begin{tabular}{@{ }p{11cm}l@{ }}
	Any diagram of $\bR \in \bfup{Trans}$ from $A$ to $B$ can be put in the form on the right for some finite set $Q$ and $\T \in \bfup{FinRel}$.&\input{graph-from-trans-fin.tikz}
	\end{tabular}
\end{proposition}

\begin{proof}
	We start by using the first equation of \Cref{fig:graph-feedback} from right to left to push all the feedbacks at the bottom of the diagram. Then, we use the equations of \Cref{fig:graph-finrel} to merge all the non-feedback into a single box. Lastly, we use the last equation of \Cref{fig:graph-feedback} to merge all the feedbacks into a single feedback.
\end{proof}
\begin{theorem}[Universality]\label{thm:universality-fin}
	For all $\bS \in \bfup{Trans}$, $\binterp{\bS}$ is a regular relation. For all regular relation $\R$, there exists $\bS \in \bfup{Trans}$ such that $\binterp{\bS} = \R$.
\end{theorem}
\begin{proof}
	This follows from \Cref{prop:normal-form-fin} and \Cref{prop:transducer-fin}. 
\end{proof}

\subsection{Simulation Principle and Completeness}\label{sec:completeness-fin}

We claim that we can complete our equational theory by adding a single equation, the simulation principle of \Cref{fig:fin-simulation-principle}.

\begin{figure}[!h]
	\[\input{fin-simulation-principle.tikz}\]
	\caption{Simulation Principle for Finite Words.}
	\label{fig:fin-simulation-principle}
\end{figure}

The core idea behind this principle is that a finite run can be see as a diagram starting with $I$, followed by a finite number of applications of $\R$ and finally ends with $F$. In such diagram, $I$ can turn into $J$ by "spawning" a $\S$, which can then move through $\R$, turning it into $\T$, and finally be absorbed by $F$ to turn it into $G$, proving the equivalence with a run in the second transducer. This principle is sound with respect to $\binterp{-}$ (this is proved formally in \Cref{app:fin-simulation-principle}). The principle directly implies the sliding rule \Cref{fig:graph-slide} allowing in particular to equate transducers with isomorphic sets of states. Notice that \Cref{fig:graph-slide} allows for a $\bR$ that might itself include some feedbacks.

\begin{figure}[h]
	\[\input{graph-slide.tikz}\]
	\caption{Sliding.}
	\label{fig:graph-slide}
\end{figure}

\begin{theorem}[Completeness]\label{thm:completeness-fin} For $\bR$ and $\bT$ two diagrams of \bfup{Trans} from $A$ to $B$. Whenever $\binterp{\bR} = \binterp{\bT}$, we can rewrite $\bR$ into $\bT$ by using only the rules of \bfup{Trans} and \Cref{fig:fin-simulation-principle}.
\end{theorem}

The core idea of the proof is diagrammatically mimic the determinization and minimization from automata theory. We recall that a finite automaton is simply a transducer with $\one$ for output alphabet, so we start by studying completeness in the case where $B = \one$.

\begin{definition}[Determinization]
	Let $(\T,A,Q,I,F)$ be a finite automaton. We write $\xrightarrow[\text{\small $\T$}]{}$ for the transitions within that automaton. Its determinization is the automata \[(\P(\T), A, \P^{\textup{acc}}(Q), \{I\}, \P^{\textup{acc}}_{\cap F \neq \varnothing}(Q))\]
	More precisely, we start by recalling that $\P(Q)$ is the set of all subsets of $Q$. We use $x,y$ for elements of $Q$, and $X$,$Y$ for elements of $\P(Q)$.  Then, we define the function\footnote{We consider functions to be a special case of relations.} $\P(\T): A \times \P(Q) \to \P(Q)$ as 
	$ \P(\T)(a,X)=\{y\in Q ~|~ \exists x\in X,~  x \xrightarrow[\text{\small $\T$}]{a} y \} $.
	We then consider the set $\P^{\textup{acc}}(Q)$ of subsets of $Q$ accessible by iteration of that function, starting from $I$. We can now restrict $\P(\T)$ to a function $A \mapsto \P^{\textup{acc}}(Q) \to \P^{\textup{acc}}(Q)$. Lastly, $\P^{\textup{acc}}_{\cap F \neq \varnothing}(Q)$ is the set of accessible subsets of $Q$ that have a non-empty intersection with $F$.
\end{definition}

\begin{proposition}\label{prop:determinization}
	For all finite automata $(\T,A,Q,I,F)$, we have
	\[ \input{fin-determinization.tikz}\]
	where $\ni : \P^{\textup{acc}}(Q) \to Q$ is the usual ``contains'' relation, that is $X \ni x$ whenever $x \in X$.
\end{proposition}
\begin{proof}
	Using the logical reasoning as in \Cref{sec:logic}, we can rewrite the equation as the following. We are looking at $\forall a\in A,~ \forall X\in \P^{\textup{acc}}(Q),~\forall y\in Q,~$
	
	\[ \left(\exists x\in Q,~ (x\in X) \land (x \xrightarrow[\text{\small$\T$}]{a} y)\right) ~\Leftrightarrow~ \left(\exists Y\in \P^{\textup{acc}}(Q),~   (X \xrightarrow[\text{\small$\P(\T)$}]{a} Y) \land (y\in Y)\right)\]
	We start by reformulating the right side of the equivalence, as $\P(\T)$ is a function, we can remove the $\exists$ and write $y \in \P(\T)(a,X)$. Then, using the definition of $ \P(\T)$, we obtain that it is equivalent to $\exists x \in X, x \xrightarrow[\text{\small$\T$}]{a} y$, which is exactly the left side of the equivalence.
\end{proof}

Rephrased in our graphical language and applyin the simulation principle, this gives:
\begin{corollary}\label{cor:determinization}
	For any finite automaton $(\T,A,Q,I,F)$, using \Cref{fig:fin-simulation-principle}, we have
	\[ \input{fin-determinization-conclusion.tikz}\]
\end{corollary}

Another key element of our completeness result is the minimization of automata.
\begin{definition}[Minimization]	
	Let $(\D,A,Q,\{i\},F)$ be a \textbf{deterministic} finite automaton where every state is accessible from the initial state. We write $\xrightarrow[\text{\small $\D$}]{}$ for the transitions within that automaton. We write  $\xrightarrow[\text{\small $\D$}]{w}$ with $w$ a word of size $n$ for its iterated transition $\xrightarrow[\text{\small $\D$}]{w_1}\dots \xrightarrow[\text{\small $\D$}]{w_n}$. Its minimization is the deterministic automata 
	$ (L_\D,A,L_{A^*},L_{\epsilon},L_{A^*}^{\ni \epsilon}) $.\\
	More precisely, we write $L_w = \{ v \in A^* \mid \exists f \in F, i \xrightarrow[\text{\small $\D$}]{wv} f\}$, and remark that whenever $L_w = L_u$, then for all $a \in A$ we also have $L_{wa} = L_{ua}$. We then define $L_{A^*} = \{ L_w \mid w \in A^*\}$ note that a single element of this set might correspond to multiple distinct $w$, and in fact $L_{A^*}$ is actually smaller or equal to $Q$ in cardinality. We defined its restriction $L_{A^*}^{\ni \epsilon} = \{ L_w \mid w \in A^*, \epsilon \in L_w\}$. Lastly, we define the transition function as $L_\D(a,L_w) =  L_{wa}$.
\end{definition}

This minimization is well known in the literature \cite{DBLP:books/daglib/0016921}, and yields the unique smallest deterministic automaton equivalent to the starting one.

\begin{proposition}\label{prop:minimization}
	For any deterministic finite automaton $(\D,A,Q,\{i\},F)$ where every state is accessible from the initial state, we have
	\[ \input{fin-minimization.tikz}\]
	where $\L : Q \to L_{A^*}$ relates $p \in Q$ to $\ell \in L_{A^*}$ whenever $\ell = \{ v \mid \exists f \in F, p \xrightarrow[\text{\small$\D$}]{v} f\}$. We note that since every state is accessible, there exists $w \in A^*$ such that the later is equal to $L_w$.
\end{proposition}

The core idea is similar to \Cref{prop:determinization}: we express both sides of the equality as equations relating transitions of the transducers, and prove their equivalence using standard argument of automata minimization.
Again the simulation principle can be applied to conclude that the behavior of an automata and its minimization are the same.

\begin{corollary}\label{cor:minimization}
	For any deterministic finite automaton $(\D,A,Q,\{i\},F)$ where every state is accessible from the initial state, using \Cref{fig:fin-simulation-principle}, we have
	\[ \input{fin-minimization-conclusion.tikz}\]
\end{corollary}

We can now provide a proof of the completeness result.

\begin{proof}[Proof of \Cref{thm:completeness-fin}]
	We start with two diagrams $\bR,\bS$ of \bfup{Trans} from $A$ to $B$ and assume $\binterp{\bR} = \binterp{\bS}$. We start by focusing on $\bR$. We combine it with the cup $\epsilon_B$ to bend its output into an input, and we then use \Cref{prop:normal-form-fin} to put the result in quasi-normal form. Then, we use \Cref{cor:determinization} followed by \Cref{cor:minimization} to obtain:
	\[ \input{fin-compl-epsilon.tikz} \quad = \quad \input{fin-compl-normal-form.tikz} \quad=\quad \input{fin-compl-minimal.tikz} \]
	For convenience, we name $(\R_M,B \times A,Q_M,I_M,F_M)$ the resulting minimal automaton. We do the same for $\S$ and write $(\S_N,B \times A,Q_N,I_N,F_N)$ for the resulting minimal automaton. Since $\binterp{\bR} = \binterp{\bS}$ and by soundness of the equations, we obtain that:
	\[ \binterp{\input{fin-compl-minimal-M.tikz}} \qquad = \qquad \binterp{\input{fin-compl-minimal-N.tikz}} \]
	Using the definition of $\binterp{-}$ and \Cref{prop:transducer-fin}, it follows that $(\R_M,B \times A,Q_M,I_M,F_M)$ and $(\S_N,B \times A,Q_N,I_N,F_N)$ recognize the same language, hence by uniqueness of the minimal automaton (see \cite{DBLP:books/daglib/0016921}) we obtain that the two automata are equal up to an isomorphism $\iota : Q_M \to Q_N$, hence using the simulation principle with this $\iota$, we obtain
	\[   \input{fin-compl-epsilon.tikz}  =  \input{fin-compl-minimal-M.tikz}  = \input{fin-compl-minimal-N.tikz}  =  \input{fin-compl-epsilon-bis.tikz} \]
	By combining with the cap $\eta_B$ to bend the input $B$ into an output, and using the the fact that $(\id_B \times \epsilon_B) \circ (\eta_B \times \id_B) = \id_B$, we obtain $\bR = \bS$.
\end{proof}

The simulation principle requires to find a well chosen relation between the state spaces of two automata for the language equality to follow.
This is reminiscent of the well studied simulation techniques in automata theory, and it is not a coincidence. 
In fact, the simulation principles also holds in the case where we only want language inclusion, simply by replacing every ``='' of \Cref{fig:fin-simulation-principle} by $\subseteq$ (for a backward-simulation) or by $\supseteq$  (for a forward-simulation). Thus, the simulation principle can be reinterpreted as asking to find a simulation which is at the same time forward and backward, the fact that finding a backward-forward simulations implies language equivalence has already appeared in the literature \cite{lynch1995forward}.

\section{Transducers on Bi-Infinite Words}
\label{sec:infinite}

We will now consider the case of transducers acting on bi-infinite words, i.e. words that are infinite both ways. 
We call those transducers $\ZZ$-transducers, and note that within them the notion of initial or final state becomes irrelevant, although one could argue that in fact all the states of those transducers are both initial and final.

\subsection{$\ZZ$-Transducers}
Similarly to the $\_^*$ lift, for $A$ finite, we consider elements of $A^{\ZZ}$ as infinite words $\dots w_{-1}w_0w_1 \dots$, and we define the lifted relation $\R^{\ZZ} : A^{\ZZ} \to B^{\ZZ}$ of $\R : A \to B$ as $\{(w,v)\mid \forall i \in \ZZ, (w_i,v_i) \in \R\}$. 
Similarly to \bfup{UniRel} with $\_^*$, we write $\ZZ$-\bfup{Rel} for the compact closed strict symmetric monoidal category of relations over sets of the form $A^\ZZ$, and reuse the diagrammatic representations of \bfup{FinRel} in this setting. Similarly to $\_^*$, we can apply $\_^{\ZZ}$ to individual wires as below, and this lifting operation is a faithful strong symmetric monoidal functor from \bfup{FinRel} to $\ZZ$-\bfup{Rel}.
\[ \input{rel-lifted-z.tikz}\]

We can also define $\ZZ$-transducers -- a.k.a transducers on bi-infinite words. They are simply tuples $(\T,A,B,Q)$ similar to transducers but without a specified sets of initial or final states. When $B=\one$, such transducers corresponds to transition systems with finite set of states. When $A=\one$, they can also be interpreted as generalized symbolic dynamic systems. Indeed the set of elements of $Q$ in relation with a symbol $b\in B$ form an overlapping partition of $Q$, and the transition function can be interpreted as the action of a non-deterministic dynamic.

We reuse the same notations as for usual transducers, and note that a \textbf{run} within an $\ZZ$-transducer is now a bi-infinite sequence $\dots \xrightarrow[b_{-1}]{a_{-1}} q_{-1} \xrightarrow[b_0]{a_0} q_0 \xrightarrow[b_1]{a_1} q_1 \xrightarrow[b_2]{a_2} \dots$. 
When such a run exists, we say that the transducer can transform the bi-infinite word $\dots a_{-1}a_0a_1\dots$ into $\dots b_{-1}b_0\dots b_1\dots$.
The \emph{behavior} of the transducer is the relation from $A^\ZZ$ to $B^\ZZ$ which relates $w$ to $v$ whenever the transducer can transform $w$ into $v$. We write $\L^\ZZ(\T,A,B,Q)$ for its behavior. 
We define the \textbf{infinite shift} $\shift{}{}_A : A^\ZZ \to A^\ZZ$ as the relation $\{ (w,v) \mid \forall k \in \ZZ, w_k = v_{k+1} \}$. 
Said otherwise, $\shift{}{}_A((a_n)_{n \in \ZZ}) = (a_{n-1})_{n \in \ZZ}$. 
We often omit the alphabet $A$ and only write $\shift{}{}$ when the alphabet is clear from context and we represent it and its transposed:
$\input{rel-shift-infinite.tikz}$. As in the finite case, we can obtain the behavior of a transducer from the shift.

\begin{proposition}\label{prop:zeta-transducer-red}
	For every $\ZZ$-transducer $(\T,A,B,Q)$, its behavior can be obtained using the lift of $(\T,A,B,Q)$ and the shift as follows:
	\[ \input{rel-trans-inf-red.tikz} \]
\end{proposition}

In \Cref{lem:zeta-transducer-from-transducer} we will link the behavior of a $\ZZ$-transducer $(\T,A,B,Q)$ to the behavior of the transducer $(\T,A,B,Q,Q,Q)$, that is, where all the states are initial and final.

\subsection{Sofic Relations}
\label{sec:subshifts}

We define the analog of regular languages and relations in the bi-infinite case.
We start with \emph{sofic subshifts} which have been extensively studied in symbolic dynamics.

\subsubsection{Sofic Subshifts, Factor-Closure and Limits}
\label{subsec:sofic-subshfits}

For $A$ finite, we can define a distance on $A^\ZZ$ as 
$\mathrm{d}(x,y) = 2^{-\inf\{i \in  \NN, x_i\neq y_i \vee x_{-i} \neq y_{-i}\}}$. The topology induced by this distance is called the prodiscrete topology (or Cantor topology), for which $A^\ZZ$ is compact and closed. Subsets of $A^\ZZ$ that are shift-invariant (meaning $\shift{}{}_A(X) = X$) and closed for the prodiscrete topology are called \textbf{subshifts}.

\begin{definition}[Factor] Given two potentially bi-infinite words $x,y \in A^* \cup A^\ZZ$, we say that $x$ is a factor of $y$ and write $x \sqsubseteq y$ whenever there exists $u$,$v$ such that $y = uxv$. Additionally, we say that is $x$ is an internal factor of depth $n$, and write $x \internal{n} y$, whenever both $u$ and $v$ are of length at least $n$ (including infinity).	
	For $x\in A^{\mathbb{Z}}$, we then define $\fact{x}:= \{ u \in A^* | u \sqsubseteq x\}$. Similarly, for $X$ a subshift or a language, we define  $\fact{X}:= \{ u \in A^* | \exists x \in X, u\sqsubseteq x\}$. A language $L$ is said \textbf{factor-closed} if $\fact{L} = L$.
\end{definition}

A subshift $X \subseteq A^\ZZ$ is said to be \textbf{sofic}\footnote{There are several equivalent definitions of sofic subshifts, we briefly discuss this in Appendix \ref{app:bijection-subshifts}} if $\fact{X}$ is regular. 
More than that, a sofic subshift can actually be generated by its factors with a well-chosen notion of limit.

For the remaining of this section, whenever we consider words of $A^*$, we consider them to be indexed as if they were words of $A^\ZZ$.
By default, $u \in A^*$ is indexed from $0$ to $|u|-1$ but we write $\shift{k}{}(u)$ for its $k$-shifted version indexed from $-k$ to $-k+|u|-1$. 
Whenever $i$ is not a valid index for $u$, we write $u_i = \bot$.
This allows us to extend\footnote{The formula is unchanged, that is $\mathrm{d}(x,y) = 2^{-\inf\{i \in  \NN, x_i\neq y_i \vee x_{-i} \neq y_{-i}\}}$; see Appendix \ref{app:bijection-subshifts}.} the distance $\mathrm{d}(u,x)$ to a finite word $u$ and an infinite word $x$.

\begin{definition}[Limit Language]
	Given a language $L\subseteq A^*$, we define its limit language $\limlang{L}$ as the set of bi-infinite words that can be reached as a limit of words in $L$, that is :
	\[ \limlang{L} := \{ x \in A^{\mathbb{Z}}, \exists (u_n)_{n\in \mathbb{N}} \in L^{\mathbb{N}}, (k_n)_{n\in \mathbb{N}} \in \mathbb{N}^{\mathbb{N}}\ s.t.\ \lim\limits_{n\to \infty} \shift{k_n}{}(u_n) = x\}\] 
\end{definition}

\begin{lemma}\label{lem:limlang-is-subshift}
	Given a language $L\subseteq A^*$, its limit language $\limlang{L}$ is a subshift. 
\end{lemma}

It is clear that different languages can have the same limit, typically a language and its factor closure have the same limit,
however there is a canonical language for a given sofic subshift.
The idea is to remove all the words that are not relevant for the limit behavior.

\begin{definition}[Pruned Language]
Given a language $L\subseteq A^*$, we define its pruned sub-language $\prun{L}$ as the set of words in $L$ that have arbitrarily long (simultaneous) left and right extensions that is : $\prun{L} := \{ u \in L, \forall n \in \mathbb{N}, \exists v \in L, u \internal{n} v\}$.
A language $L$ is said to be \textbf{pruned} when $L=\prun{L}$, or equivalently whenever for every $w \in L$ there exists two non-empty words $u,v$ such that $uwv \in L$.
\end{definition}

\begin{proposition}\label{prop:bijection-subshifts}
	$\fact{-} : \P(A^\ZZ) \to \P(A^*)$ can be restricted to a bijection between sofic subshifts and pruned factor-closed regular languages, and has $\limlang{-} : \P(A^*) \to \P(A^\ZZ)$ for inverse when restricted similarly.
\end{proposition}

This result allows to reduce questions on sofic subshifts to questions on pruned factor-closed regular languages.

\subsubsection{Presentation of a Sofic Subshift}

We have established a bijection between sofic subshifts and factor closed-regular languages, now building on this we introduce a normal form for automata that allows us to reproduce our completeness result to the infinite case.

\begin{definition}[Presentation of a Sofic Subshift]
	A \textbf{presentation} $(\T,A,Q)$ of a sofic subshift $X$ is an automaton $(\T,A,Q,Q,Q)$, meaning  all the  states are initial and final, recognizing a language $L$ and such that $\limlang{L}=X$. A presentation is said 
	\begin{description}
		\item[Pruned] if its language is pruned, meaning that for every word $w$ in the language, there exists $u$ and $v$ both non-empty such that $uwv$ is in the language.
		\item[Right-Resolving] whenever its transition relation is a partial function.
		\item[Rooted] whenever there exists a root state $r$ such that
		\begin{itemize}
			\item every accepted word admits an accepted run starting by $r$,
			\item and every state is accessible from $r$, meaning there is a run from $r$ to that state.
		\end{itemize}
	\end{description}
\end{definition}

\begin{theorem}[Uniqueness]\label{thm:minimal-uniqueness-inf}
	Every non-empty\footnote{If the sofic subshift is empty, its only pruned presentation is the empty automaton, but this automaton is not rooted.} sofic subshift admits a unique minimal rooted right-resolving pruned presentation, up to isomorphism.
\end{theorem}

Notice that this result is not in contradiction with the well known fact that the minimal presentation of a sofic subshift is not unique in general \cite{lind-marcus}. Indeed, it is often possible to find non-rooted presentations that are smaller than our ``minimal'' one, but it is from this rootedness that we obtain uniqueness.

\subsubsection{Sofic Relations and $\ZZ$-Transducers}\label{sec:sofic-z-transducer}

A relation $\R : A^\ZZ \to B^\ZZ$ is said sofic if it is a sofic subshift when seen as a subset of $(A\times B)^\ZZ$, which we know is equivalent to having a representation $(\T,A \times B,Q)$. Using the usual equivalence between automaton on $A \times B$ and transducers from $A$ to $B$, and a relation  $\R : A^\ZZ \to B^\ZZ$ is sofic if and only if there exists a $\ZZ$-transducer $(\T,A,B,Q)$ such that $\limlang{\L(T,A,B,Q,Q,Q)} = \R$. Using the following lemma, $\R$ is sofic if and only if it is the behavior of a $\ZZ$-transducer.

\begin{lemma}\label{lem:zeta-transducer-from-transducer}
	Given a $\ZZ$-transducers $(\T,A,B,Q)$, we consider the transducer  $(\T,A,B,Q,Q,Q)$ where all the states are initial and final.
        We have $\L^\ZZ(\T,A,B,Q) = \limlang{\L(\T,A,B,Q,Q,Q)}$.
\end{lemma}

The composition or product of two sofic relations is sofic, one can simply build the associated $\ZZ$-transducer using the exact same construction as in \Cref{prop:RegRel-is-cat}, so we write \bfup{SofRel} for the subcategory of $\ZZ$-\bfup{Rel} of sofic relations.

\subsection{Diagrams for Subshifts and Completeness}\label{sec:diag-shift}

\begin{wrapfigure}{r}{0.35\textwidth}\vspace{-0.2cm}
\input{graph-fin-to-inf.tikz}
\end{wrapfigure}
The graphical language $\ZZ$\bfup{-Trans} is a variant of \bfup{Trans}. 
Its definition and equational theory (without simulation principle) are exactly the same, up to changing \blue{\textbf{thick blue}} into \red{\textbf{thick red}} and applying the substitution on the right 

We refer to \Cref{app:infinite-language} for an explicit definition. 
Adapting the simulation principle is more complex, so we will treat it explicitly in the core of this paper. 
We write $\rinterp{-}$ for its semantics, going from  $\ZZ$\bfup{-Trans} to $\ZZ$\bfup{-Rel}, which is simply ``removing the color and adding a $\_^\ZZ$ everywhere''. The explicit definition of this semantics is given in \Cref{appfig:inf-interp}. 
As in the finite case, even without a simulation principle, we have a quasi-normal form and our language is universal for sofic relations, see \Cref{app:inf-quasi-normal-form,app:inf-universality} for the exact statement.

Adapting the simulation principle to this infinite case is non-trivial, as we cannot simply remove the preconditions that acted on $I$ and $F$ and expect the equation to be sound. We propose in \Cref{fig:inf-simulation-principle} some relatively complex preconditions. 
\begin{figure*}[!h]
	\[\input{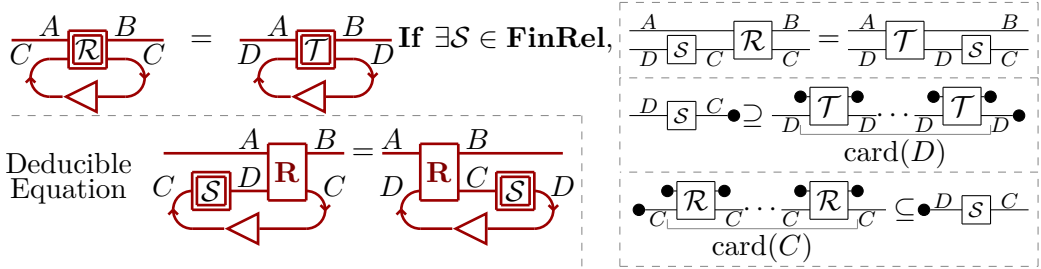}\]
	\caption{Simulation Principle for Bi-Infinite Words, and an Equation Deducible from it.}
	\label{fig:inf-simulation-principle}
\end{figure*}

Informally, the second precondition means that any state of  the $\ZZ$-transducer $(\T,A,B,D)$ that is the start of a path of size $\textup{card}(D)$ (hence a loop must exists, hence it is the start of an infinite path) must be in the domain of $\S$. Conversely, the third precondition means that any state of the $\ZZ$-transducer $(\R,A,B,C)$ that is the end of a path of size $\textup{card}(C)$ (hence a loop must exists, hence it is the end of an infinite path) must be in the codomain of $\S$. We note that most of the time this equation is used in the case where $\S$ is a total and surjective relation (that is $\S^{-1}(C) = D$ and $\S(D) = C$), which automatically satisfies the second and third preconditions.
Like in the finite case, we can derive from the simulation principle one of the well-known equations of trace monoidal categories.
The only missing equation to obtain a trace monoidal category is that adding a feedback loop to $\gamma_{A,A}$ yields the identity $\id_A$ (which would require to remove the ``shift'' from our feedback).
We can now state the completeness theorem.

\begin{theorem}[Completeness]\label{thm:completeness-inf} For $\rR$ and $\rT$ two diagrams of $\ZZ$\bfup{-Trans} from $A$ to $B$. If $\rinterp{\rR} = \rinterp{\rT}$, we can rewrite $\rR$ into $\rT$ by using only the rules of $\ZZ$\bfup{-Trans} and \Cref{fig:inf-simulation-principle}.
\end{theorem}

While the proof still relies on the uniqueness of minimal deterministic automatons, two differences are notable: (1) Our ``automaton'' is now a presentation of a sofic subshift, meaning that we rely on \Cref{thm:minimal-uniqueness-inf} instead of the usual theorem about uniqueness of the minimal automaton. (2) Before the determinization step, we need to prune all the states of the presentation that are not part of any bi-infinite path.

\section{Conclusion}
\label{sec:conclusion}

The diagrammatical axiomatization presented in this paper presents many similarities with recent works also relying on feedback loops to represent internal state space.
There was a thread of work aiming to handle diagrammatically streams of data, first \cite{sprunger2019differentiable} in the classical case and then \cite{di2022monoidal} and \cite{carette2021graphical} in the probabilistic and quantum cases.
Our simulation principle bears similarity with equations appearing in \cite{ghica2022fully}.
In \cite{piedeleu2021string}, the authors provide a complete diagrammatical theory for finite automata with similar goals, however the structure or the categories involved is very different as they do not use the cartesian product as monoidal structure.
However, we do not know of any attempt to handle subshifts with such formalism.
Unifying transducers and sofic subshifts with similar graphical languages opens many research directions towards diagrammatical proofs in symbolic dynamics. We present a few of them.

\textbf{Subshifts on $\mathbb{N}$} A straightforward next step would be to consider the situation in-between finite and bi-infinite words. Indeed, this would allow to draw interesting bridges between symbolic dynamics on the half-line, Buchi automata and LTS with initial states. This direction would also allow easier comparison with the monoidal stream line of research. 

\textbf{One-dimensional SFTs} One of the most natural objects of symbolic dynamics are subshift of finite type (SFTs), defined by a finite number of local constraints.
Surprisingly, their diagrammatical definition is less straightforward than sofic ones.
The logical next step would be to define SFTs diagrammatically and use this language to show classical results linking sofic subshifts, SFTs and cellular automata.

\textbf{Symbolic dynamics on groups} Symbolic dynamics also studies infinite colorings of $\ZZ^2$ and more generally Cayley graphs of finitely generated groups. 
Diagrammatical proofs seem promising to tackle this since they seem to abstract the topology of the graph, "hidden" in the shift relation.
This would suggest that some diagrammatical one-dimensional proof techniques may be adapted to higher dimensions and groups quite straightforwardly.

\textbf{General induction principle} Many objects in symbolic dynamics have a coinductive definition (substitutive subshifts, limit sets of cellular automata, ...). 
Generalizing our simulation principle to more abstract categories could provide new proofs techniques usable for all these co-inductive objects in an abstract way.

Finally, abstracting our diagrams to other settings/categories than relations may also allow us to study other problems such as counting subshifts, studying ergodic probability measures, and defining a notion of quantum subshifts.

\printbibliography

\clearpage
\appendix
\addtocontents{toc}{\protect\setcounter{tocdepth}{0}}
\section{The Category of Relations}
\label{app:relations}

\begin{figure*}
	\input{rel-equations.tikz}
	\caption{Equations for a Strict Symmetric Monoidal Category.}
	\label{fig:rel-equations}
\end{figure*}
\begin{figure*}
\input{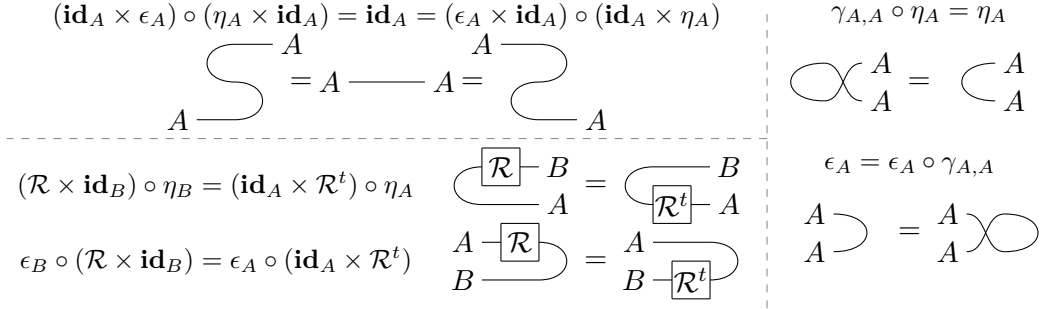}
\caption{Equations for the Cup and Cap}
\label{appfig:rel-equations-cup-cap}
\end{figure*}

All the diagrams of \Cref{sec:relations} can actually be formalized using string diagrams from category theory \cite{MacLane}. Indeed, finite sets and relations is well known to be a strict symmetric monoidal category, called \bfup{FinRel}, that is:
\begin{itemize}
	\item A set of objects (finite sets).
	\item A binary operation on objects (the Cartesian product $\times$) which is associative and has a neutral element (\one).
	\item A set of morphisms between those objects (relations), a way to compose them (the usual composition $\circ$) that is associative and has a neutral element (the identity $\id$).
	\item A binary operations on morphisms (the Cartesian product of relations $\times$) which is associative, has a neutral element ($\id_\one$), and is a bifunctor (see the last equation of \Cref{fig:rel-equations}).
	\item For every two objects $A$ and $B$, a morphism that swaps them ($\gamma_{A,B}$) which forms a natural isomorphism (third and fourth equations of \Cref{fig:rel-equations}).
\end{itemize}
We list all the resulting equations in \Cref{fig:rel-equations}, both in text and diagrams. In those diagrams, gray boxes correspond to ``bracketing'', and the overall consequence of those equations is that we can keep the bracketing implicit.

If we account for the cup and cap, relations form a compact closed category (where additionally all the objects are self-duals) \cite{KellyLaplaza}, which correspond to the additional equations listed in \Cref{appfig:rel-equations-cup-cap}. 

\section{Proofs for Finite Words}
\label{app:finite}

\subsection{Composing Transducers}\label{app:transducer-category}

In this section, we prove \Cref{prop:RegRel-is-cat} by showing that we can compose and make a product of transducers. 

Given two transducers $(\T,A,B,Q,I,F)$ and $(\S,C,D,P,J,G)$, we can build their product as $(\R,A \times C,B \times D,Q \times P, I \times J, F \times G)$ where $((a,c),(q,p))~\R~((b,d),(q',p'))$ whenever $(a,q)~\R~(b,q')$ and $(c,p)~\R~(d,p')$. Its behavior is exactly the product of the two behaviors. 
Assuming $B=C$, we can build their composition as $(\U,A,D,Q \times P,I \times J, F \times G)$ where  $(a,(q,p))~\U~(d,(q',p'))$ whenever there exists $b \in B$ such that  $(a,q)~\R~(b,q')$ and $(b,p)~\R~(d,p')$.	 Its behavior is exactly the composition of the two behaviors.

\subsection{From Transducers to Graphical Representation}\label{app:transducer-fin}

In this section, we prove \Cref{prop:transducer-fin}, that is for every transducer $(\T, A, B, Q, I,F)$, its behavior can be obtained by lifting $\T$ and using the shift as follows:
\[ \input{rel-trans-fin.tikz} \]

	We write $\R$ for the left-hand-side. We take $w \in A^*$ and $v \in B^*$ of size $n$, and note that by definition of the composition of relations, $w~\R~v$ if and only there exists two words $u,t \in Q^*$ such that $(w,u)~\T^*~(v,t)$ and $u~\transpose{\left(\shift{I}{F}\right)}~t$. 
	
	Let us focus on the former, it is equivalent to, for every $1 \leq k \leq n$, $(w_k,u_k)~\T~(v_k,t_k)$, that is $u_k \xrightarrow[v_k]{w_k} t_k$.
	
	Let us focus on the latter, it is equivalent to the first letter of $u$ being in $I$, the last letter of $t$ being in $F$, and for all $1 \leq k \leq n$, $u_{k+1} = t_k$. Said otherwise, there exists a sequence $q_0,q_1,\dots,q_n \in Q$ such that $u = q_0\dots q_{n-1}$, $t = q_1\dots q_n$, $q_0 \in I$ and $q_n \in F$.
	
	Coming back to the initial equivalence, we have $w~\R~v$ if and only if there exists a sequence $q_0,q_1,\dots,q_n \in Q$ such that $q_0 \in I$, $q_n \in F$, and for all $1 \leq k \leq n$, $q_{k-1} \xrightarrow[v_k]{w_k} q_k$. Since this is exactly the definition of a run, we have that $w~\R~v$ if and only if $w~\L(\T,A,B,Q,I,F)~v$.

\subsection{Soundness of the Standard Equations}\label{app:fin-soundness}

\begin{figure*}[h]
	\input{graph-generators.tikz}
	\caption{Generators of \bfup{Trans}.}
	\label{appfig:graph-generators}
\end{figure*}

\begin{figure*}[h]
	\input{rel-equations-fin.tikz}
	\caption{Equations for a Strict Symmetric Monoidal Category.}
	\label{appfig:rel-equations-fin}
\end{figure*}

\begin{figure*}[h]
	\input{graph-feedback.tikz}
	\caption{Equations for a Feedback Category.}
	\label{appfig:graph-feedback}
\end{figure*}

\begin{figure*}[h]
	\input{graph-finrel.tikz}
	\caption{Equations for Faithfully Embedding \bfup{FinRel}.}
	\label{appfig:graph-finrel}
\end{figure*}

\begin{figure*}
	\[\input{fin-interp.tikz}\]
	\caption{Inductive Definition of the Semantics $\binterp{-} : \bfup{Trans} \to \bfup{UniRel}$.}
	\label{appfig:fin-interp}
\end{figure*}

We start by  recalling the generators of our language and all its equations in \Cref{appfig:graph-generators,appfig:rel-equations-fin,appfig:graph-feedback,appfig:graph-finrel,appfig:graph-finrel}. Then we formally define the semantics $\binterp{-} : \bfup{Trans} \to \bfup{Rel}$ inductively on the syntax, as shown in \Cref{appfig:fin-interp}.

We now prove the soundness of the equational theory. The soundness of the equations stating that $\bfup{Trans}$ is a strict symmetric monoidal category follows immediately from the fact that $\bfup{Rel}$ is a strict symmetric monoidal category, so we only need to look at the equations of \Cref{fig:graph-feedback}. The top-left and bottom-left equations also follow from the fact that $\bfup{Rel}$ is a strict symmetric monoidal category. The top-right is sound because $\shift{\one}{\one}_\one = \id_\one$, and the bottom-right is sound because $\shift{I}{F}_C \times \shift{J}{G}_D = \shift{I \times J}{F \times G}_{C \times D}$. \qed

\subsection{Soundness of the Simulation Principle}\label{app:fin-simulation-principle}

\begin{figure*}
	\[\input{fin-simulation-principle-black.tikz}\]
	\caption{Semantics of the Simulation Principle for Finite Words.}
	\label{appfig:fin-simulation-principle-black}
\end{figure*}

We prove the soundness of the equation of \Cref{fig:fin-simulation-principle} with respect to $\binterp{-}$, that is, we want to prove the equation of \Cref{appfig:fin-simulation-principle-black}. 

Since all the relations involved are uniform, we can make a case-by-case analysis on the size of the input words on $A^*$. This mean we are now trying to prove the following for all $n \geq 0$:
\[\input{fin-simulation-principle-sound.tikz}\]
Then, using the hypothesis of \Cref{appfig:fin-simulation-principle-black}, we can rewrite $I$ into $\S \circ J$, then move the $\S$ through the $n$ copies of $\R$ to transform them into $n$ copies of $\T$, and lastly rewrite $\S \circ F$ into $G$.

\begin{figure*}[h]
	\[\input{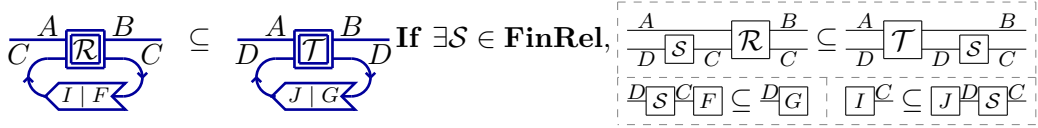}\]
	\caption{Backward-Simulation Principle for Finite Words.}
	\label{appfig:fin-backward-simulation-principle}
\end{figure*}

\begin{figure*}[h]
	\[\input{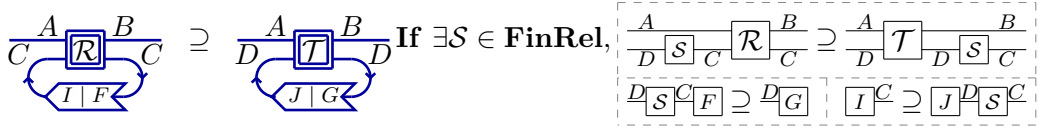}\]
	\caption{Forward-Simulation Principle for Finite Words.}
	\label{appfig:fin-forward-simulation-principle}
\end{figure*}

In the above proof, one could replace all the instances of $=$ by $\subseteq$, or all of them by $\supseteq$, this would yield a proof of soundness for the backward-simulation principle of \Cref{appfig:fin-backward-simulation-principle} or of the forward-simulation principle of \Cref{appfig:fin-forward-simulation-principle}. \qed

\subsection{Minimization of finite transducers}\label{app:prop:minimization}
	We provide the full proof of \Cref{prop:minimization}.

	Using the logical reasoning as in \Cref{sec:logic}, we can rewrite the equation as the following. We are looking at $\forall a \in A, \forall p \in Q, \forall \ell \in L_{A^*},$
	\[ \begin{array}{c} \left(\exists q\in Q,~  (p \xrightarrow[\text{\small$\D$}]{a} q)\land (\ell = \{ w \mid \exists f \in F, q \xrightarrow[\text{\small$\D$}]{w} f\}) \right) \\ \Leftrightarrow \\ \left(\exists m \in L_{A^*},~  (m = \{ v \mid \exists f \in F, p \xrightarrow[\text{\small$\D$}]{v} f\}) \land (m \xrightarrow[\text{\small$L_\D$}]{a} \ell)\right) \end{array} \]

	Before studying either side of that equivalence, we note that since $p$ is accessible from the initial state, there exists a word $u \in A^*$ such that $i \xrightarrow[\text{\small$\D$}]{u} p$, hence $L_u = \{ v \mid \exists f \in F, p \xrightarrow[\text{\small$\D$}]{v} f\}$.

	We start by reformulating the second part of the equivalence.  Since the formula for $m$ is given, we can remove the $\exists m$ and obtain $L_u \xrightarrow[\text{\small$L_\D$}]{a} \ell$. We can now use the definition of $L_\D$ and simplify the later in $\ell = L_{ua}$.

	We now look at the first part of the equivalence, since $\D$ is deterministic, $q$ is uniquely determined so we can remove the $\exists q$ and we obtain $\ell = \{ w \mid \exists f \in F, \D(a,p) \xrightarrow[\text{\small$\D$}]{w} f\}$. Since  $i \xrightarrow[\text{\small$\D$}]{u} p$, this is equivalent to $\ell = \{ w \mid \exists f \in F, i \xrightarrow[\text{\small$\D$}]{uaw} f\}$, that is $\ell = L_{ua}$.

	\qed

\subsection{Deducing the Sliding Equation}\label{app:fin-slide}

\begin{figure*}
	\[\input{graph-slide.tikz}\]
	\caption{Deducible Equation.}
	\label{appfig:graph-slide}
\end{figure*}

We prove that the equation from \Cref{appfig:graph-slide} can be deduced from the others. For that, we start by using \Cref{prop:normal-form-fin} on $\bR$:
\[\input{graph-slide-proof-normal-form.tikz}\]
We can then merge the two feedbacks into one, rewriting both sides of the equation we are trying to prove into:
\[\input{graph-slide-proof-rewrite.tikz}\]

We then conclude using the the simulation principle, using the following prerequisites:
\[\input{graph-slide-proof-prerequisites.tikz}\]
~\qed

\section{Proofs for the Sofic Subshifts}
\label{app:subshifts}

\subsection{Sofic subshifts and pruned factor-closed language}
\label{app:bijection-subshifts}
Here we present the proofs that where omitted in \Cref{sec:subshifts} and a few additional minor or intermediate results on the correspondence between sofic subshifts and pruned factor-closed languages.

First recall that, in Section \ref{sec:subshifts}, we define sofic subshifts as the subshifts whose factor language is regular.
This may differ from the definition with which the reader is familiar with, as a more usual definition is that a susbshift is sofic when it can be defined as the image of a SFT (subshift of finite type) by a cellular automaton.

Note that we deal here only with symbolic dynamics in dimension 1, that is the full shift is $A^\ZZ$ for some finite alphabet $A$.
Hence another definition of a sofic shift is a subshift $X$ which can be defined by a regular language $\mathcal{F}$ of forbidden words. Denote $A^*\cdot \mathcal{F}\cdot A^*$ the set of finite word containing a word in $\mathcal{F}$ as a factor subword.
Now remark that the factor language of $X$ is precisely the complement of $A^*\cdot \mathcal{F} \cdot A^*$, since $\mathcal{F}$ is regular then $\fact{X}$ is also regular.
Conversely if $X$ is a subshift with regular factor language $\L$, then it can be defined as avoiding the complement of $\L$ which is also a regular language.

For the same considerations (though with a different terminology) one may refer to \cite{lind-marcus}.

We can refine the factor relation $x \sqsubseteq y$ into $x \subset y$ by asking for the indexes to match. More precisely, for $u$ a finite word and $x$ a bi-infinite word, we write $\shift{k}{}(u)\subset x$ whenever $u=x_{-k}\dots x_{-k+|u|-1}$. We remark that $\mathrm{d}(x,\shift{k}{}(u))= \sup\{ \mathrm{d}(x,y) | y\in A^\ZZ\ s.t.\ \shift{k}{}(u)\subset y\}$.

\begin{remark}[Comparing finite and infinite words]
  For simplicity we extend the distance $\mathrm{d}$ to compare finite words and infinite words.
  However this does not define a metric on the mixed space $A^*\cup A^\mathbb{Z}$.

  The most rigorous approach to talk of a convergence of a sequence of finite words to an infinite word is that of cylinders.

  The \textbf{cylinder} of word $u$ at position $k$, denoted by $[u]_k$, is the set of all bi-infinite words that coincide with $u$ at position $k$ that is
  \[ [u]_k := \{ x \in A^\ZZ | \forall 0\leq i <|u|, x_{i-k}=u_i\} \]

  In our notation this corresponds to $[u]_k = \{ x \in A^\ZZ | \shift{k}{}(u) \subset x\}$.

  We call \textbf{domain} of a cylinder $[u]_k$ the integer interval $\llbracket -k, |u|-k-1\rrbracket$, and we denote it here by $\mathcal{D}([u]_k)$.

  Given a family $(C_n)_{n\in\NN}$ of cylinders, if $\mathcal{D}(C_n)\underset{n\to\infty}{\longrightarrow} \ZZ$, then there exists a sub-family $(C_{f(n)})_{n \in \NN}$, for each $n$ a larger cylinder $C_{f(n)} \subset C_n'$ and an infinite word $x$ such that $\{x\} = \bigcap\limits_{n\in\NN} C_n'$.

  In particular if $\mathcal{D}(C_n)\underset{n\to\infty}{\longrightarrow} \ZZ$ and for all $n$, $C_{n+1}\subseteq C_n$ then $\bigcap\limits_{n\in\NN} C_{n}=\{x\}$ for some bi-inifinite word $x$.

  The easiest case is when each cylinder is some $C_n=[u_n]_{k_n}$ such that $u_n \internal{1} u_{n+1}$ with $u_{n+1}= v\cdot u_n \cdot v'$, $|v|=k_{n+1}-k_{n}>0$ and $|v'|>0$.
  In this case we indeed have $\bigcap\limits_{n\in\NN} C_{n}=\{x\}$ which, outside this remark, we write $\shift{k_n}{}(u_n) \underset{n\to\infty}{\longrightarrow} x$.
\end{remark}

\begin{remark}[Intersection of cylinders compared to limit of words]
  The notion of limit of words is a bit more permissive than that of intersection of cylinders.
  Take for example the sequence of word $u_n := 1\cdot 0^{2n+1}\cdot 1$ and positions $k_n=n+1$.
  We have $\shift{k_n}{}(u_n) \underset{n \to \infty}{\longrightarrow} 0^\ZZ$, but $\bigcup\limits_{n\in\NN}[u_n]_{k_n} = \emptyset$.
  However, taking $v_n = 0^{2n+1}$ and $k_n'=n$, we have $[u_n]_{k_n} \subset [v_n]_{k_n'}$ as $u_n = 1\cdot v_n \cdot 1$ and $\bigcap\limits_{n\in \NN}[v_n]_{k_n'} = \{ 0^\ZZ\}$.
\end{remark}

\medskip

We now prove Lemma \ref{lem:limlang-is-subshift} which states that any limit language $\limlang{L}$ is a subshift.

\begin{proof}[Proof of Lemma \ref{lem:limlang-is-subshift}]
	Let $L\subseteq A^*$.
	First remark that $\limlang{L}$ is straightforwardly shift invariant, because if $x\in\limlang{L}$ then by definition $x$ is some limit of $\shift{k_n}{}(u_n)$ and simply shifting the sequence of indices by $k$ proves that $\shift{k}{}(x)$ is in $\limlang{L}$ as $\shift{k+k_n}{}(u_n)\underset{n\to\infty}{\longrightarrow} \shift{k}{}(x)$.
	
	Now we prove similarly that $\limlang{L}$ is closed. Let $y\in A^{\mathbb{Z}}$ be the limit of a sequence $(x_n)_{n\in\mathbb{N}}$ of infinite words in $\limlang{L}$.
	First let us remark that the sequence of central factors of $x_n$ also converges to $y$, that is $(x_n)_{-n}\dots (x_n)_{n} \underset{n\to\infty}{\longrightarrow} y$.
	Since $x_n\in \limlang{L}$, it is the limit of some sequence $(\shift{k_{n,m}}{}{u_{n,m}})_{m\in\mathbb{N}}$ where the $k_{n,m}$ are integers and $u_{n,m}$ words in $L$. Hence for each $n$  there exists $f(n)$ such that $\shift{k_{n,f(n)}}{}(u_{n,f(n)})$ contains $(x_n)_{-n}\dots (x_n)_{n}$ and we now have that $\shift{k_{n,f(n)}}{}(u_{n,f(n)}) \underset{n\to\infty}{\longrightarrow} y$. Hence $y$ is in $\limlang{L}$ and so $\limlang{L}$ is closed.
\end{proof}

\begin{lemma}
	A regular language is factor-closed if and only if it is recognized by a finite automaton where all states are initial and final. Furthermore, all the states of the minimal automaton of a factor-closed language are final.
\end{lemma}
\begin{proof}
	First given an automata $\mathcal{A}$ where all states are initial and final and a word $u$ recognized by $\mathcal{A}$, from an accepting run of $u$ in $\mathcal{A}$ we can directly read accepting runs for all factors of $u$, so the language recognized by $\mathcal{A}$ is factor-closed.
	
	Now, we consider a factor closed regular language $L$.
	There exist a minimal deterministic automaton $\mathcal{A}$ recognizing $L$.
	As $L$ is factor closed and $\mathcal{A}$ is determistic, every accessible and co-accessible state (\emph{i.e.}, on a path from an initial state to an accepting state) is accepting.
	By minimality, all states are accessible and co-accessible so all states are accepting.
	Now considering $\mathcal{A}'$ where all states are also initial states.
	$\mathcal{A}'$ also recognizes $L$. First we have that $L$ is included in the language recognized by $\mathcal{A}'$ because any accepting run in $\mathcal{A}$ is also accepting in $\mathcal{A}'$.
	Now we prove that the recognized language is exactly $L$. 
	Let $u$ be a word accepted by $\mathcal{A}'$, it is accepted by a path starting on some state $q$, as $q$ is accessible in $\mathcal{A}$ there exist a word $v$ such that $v\cdot u$ is accepted on a path starting from an initial state $q_0$ of $\mathcal{A}$. Therefore $v\cdot u$ is in $L$, by factor closure $u$ is also in $L$.
\end{proof}

\begin{lemma}[Limit and pruning]
	Given a language $L\subseteq A^*$, we have $\fact{\limlang{L}} = \fact{\prun{L}}$.
	\label{lemma:fact-prun}
\end{lemma}
This result is quite natural as factors of limit words are precisely the words that have arbitrarily long left-and-right extensions in the language.
\begin{proof}
	We prove this lemma by double inclusion. Let $L$ be a language.
	First we prove $\fact{\limlang{L}}\subseteq \fact{\prun{L}}$.
	Let $v \in \fact{\limlang{L}}$, that is there exists $x \in \limlang{L}$ such that $v\sqsubseteq x$, that is, there exists $k\in \mathbb{Z}$ $\shift{k}{}(v)\subset x$.
	By definition of $\limlang{L}$, there exists two sequences $(u_n)_{n\in\mathbb{N}}$ of words in $L$ and $(k_n)_{n\in\mathbb{N}}$ of integers such that $\lim\limits_{n\to \infty} \shift{k_n}{}(u_n) = x$.
	Denote $k' = |k| + |v|$, by definition for any integer $l$ and for any word $v'$ such that $d(x,v')\leq 2^{-k'-l}$ we have $x_{[-k'-l,k'+l]} \sqsubseteq v'$ and since $v\sqsubseteq x_{[-k',k']}$ we have $v \sqsubseteq v'$ and the existence of $v_p, v_s$ of length at least $l$ such that $v_p v v_s =v'$.
	In particular as  $(\shift{k_n}{}(u_n))$ tends to $x$, there exists $n_0$ such that
	$d(\shift{k_{n_0}}{}(u_{n_0}),x)\sqsubseteq 2^{-k'}$ so that $v\sqsubseteq u_{n_0} \in L$.
	Now repeat the same reasoning with $u_{n_0}$, and denote $k'' = |k_{n_0}| + |u_{n_0}|$.
	As $(\shift{k_n}{}(u_n))$ tends to $x$, for any $l$ there exists $n$ such that we have $d(\shift{k_n}{}(u_n),x)\sqsubseteq 2^{-k''-l}$ so that $u_{n_0}\sqsubseteq u_n$ with a prefix $u_p$ and a suffix $u_s$ of length at least $l$ such that $u_p u_{n_0} u_s = u_n$.
	That is $v\sqsubseteq u_{n_0}$ and $u_{n_0} \in\prun{L}$ so that $v\in \fact{\prun{L}}$.
	
	Now the second inclusion.
	Let $u\in \fact{\prun{L}}$, that is $u\sqsubseteq u' \in \prun{L}$.
	For simplicity assume that $|u'|$ is odd and denote $k' = \frac{|u'|-1}{2}$.
	That is, for any $n$ there exists $u_{p,n},u_{s,n} \in A^{\geq n}$ such that $u_{p,n}u'u_{s,n} \in L$.
	Denote $k_n = |u_{p,n}| + k'$ and $u_n = u_{p,n} u' u_{s,n}$.
	By compactness of $A^{\mathbb{Z}}$, $(\shift{k_n}{}(u_n))$ admits a converging subsequence. Denote $x$ its limit, which by definition is in $\limlang{L}$.
	Since $u'$ is at the center of each $\shift{k_n}{}(u_n)$, more precisely $\shift{k_n}{}(u_n)|_{[-k', k']}=u'$, so we have $u'\sqsubseteq x$ and more precisely $x|_{[-k',k']}=u'$. So $u\sqsubseteq u' \sqsubseteq x$ and $u \in \fact{\limlang{L}}$.
\end{proof}

\begin{definition}[Orbit, closure]
	Given a set $X\subseteq A^{\mathbb{Z}}$ we define its orbit $\orbit{X}$ as the set of all possible shifts of words in $X$, that is $\orbit{X}:= \{ \shift{k}{}(x) , k\in \mathbb{Z}, x \in X\}$.
	
	We also define its orbit-closure $\closure{\orbit{X}}$ as the topological closure of its orbit.
\end{definition}

As subshifts are precisely sets of bi-inifinte words that are shift-invariant and closed, a set $X\subseteq A^{\mathbb{Z}}$ is a subshifts if and only if $X= \closure{\orbit{X}}$. And more generaly $\closure{\orbit{X}}$ is the smallest subshift containing $X$.
\begin{lemma}
	For any set $X$ of infinite words we have $\limlang{\fact{X}} = \closure{\orbit{X}}$.
	\label{lemma:lim-factor}
\end{lemma}

\begin{proof}
	Let $y \in \closure{X}$, by definition there exists a sequence $(x_n)_{n \in \mathbb{N}}$ of elements of $\orbit{X}$ such that $x_n \underset{n\to\infty}{\longrightarrow}y$, by definition there exists for each $n$, $x_n'\in X$ and $k_n\in \mathbb{Z}$ such that $x_n = \shift{k_n}{}(x_n')$ .
	Denote $\shift{n}{}(u_n):= x_n|_{[-n,n]}$, that is the central word on length $2n+1$ in $x_n$.
	We have for each $n$, $u_n$ in $\fact{X}$ as $u_n \sqsubseteq x_n = \shift{k_n}{}(x_n')$, and we have by definition that $\shift{n}{}(u_n) \underset{n\to\infty}{\longrightarrow} y$ so that $y\in \limlang{\fact{X}}$.
	
	Conversely, let $y \in \limlang{\fact{X}}$. There exists a sequence $(u_n)_{n\in\mathbb{N}}$ of words in $\fact{X}$ and a sequence $(k_n)_{n\in \mathbb{N}}$ of offsets such that $\shift{k_n}{}(u_n) \underset{n\to\infty}{\longrightarrow} y$.
	By definition for each $u_n$ there exists a $x_n$ in $\orbit{X}$ such that $\shift{-|u_n|/2}{}(u_n) \subset x_n$.
	By definition of the distance on $A^{\mathbb{Z}}$ the sequence $x_n$ also converges and to the same limit $y$. Hence $y$ is in $\closure{\orbit{X}}$.  
\end{proof}

\begin{lemma}
  \label{lemma:fact-is-pfc}
  For any set $X \subseteq A^\ZZ$, $\fact{X}$ is pruned and factor-closed.
\end{lemma}
\begin{proof}
  Let $X\subseteq A^\ZZ$ and $L=\fact{X}$.

  $L$ is factor closed, indeed for any $u\in L$, there exists $x\in X$ such that $u\sqsubseteq x$ by definition of $\fact{-}$. For any $v\in L$ we have $v\sqsubseteq u \sqsubseteq x$ so $v\in \fact{x} \subset \fact{X}=L$.
  So we have $L=\fact{L}$.
  
  $L$ is pruned, indeed for any $u\in L$, there exists $x\in X$ and $k\in \mathbb{Z}$ such that $\shift{k}{}(u) \subset x$.
  Denote $k_n = k+n$ and $u_n = x_{-k-n}\dots x_{-k+|u|+n}$.
  For each $n$ we have $u_n\in \fact{x}\subseteq L$ and we have $u\internal{n}u_n$. Hence $u\in \prun{L}$ and $L$ is pruned.
\end{proof}

We now prove Proposition \ref{prop:bijection-subshifts} which states that $\fact{}:\P(A^\ZZ) \to \P(A^*)$ is, when restricted from sofic subshifts to pruned factor-closed regular languages, is a bijection.

\begin{proof}[Proof of Proposition \ref{prop:bijection-subshifts}]
  In this proof we denote $\soficshifts\subseteq \P(A^\ZZ)$ the set of sofic subshifts, $\regularlanguages\subseteq \P(A^*)$ the set of regular languages and $\regularfactorpruned\subseteq \regularlanguages$ the set of factor-closed and pruned regular languages.

  Let $L\in\regularfactorpruned$, then by \Cref{lemma:fact-prun}, $\fact{\limlang{L}} = \fact{\prun{L}} = L$.

  Let $X\in\soficshifts$, by \Cref{lemma:lim-factor}, $\limlang{\fact{X}} = \closure{\orbit{X}} = X$ and since $X$ is sofic, $\fact{\limlang{\fact{X}}} = \fact{X}$ is regular. It is also pruned and factor-closed by \Cref{lemma:fact-is-pfc}.

  Therefore the restriction $\fact{-} : \soficshifts \to \regularfactorpruned$ is bijective and its inverse is the restriction $\limlang{-} : \regularfactorpruned \to \soficshifts$.
\end{proof}

\subsection{Characterization of the Behavior of a Transducer}\label{app:transducer-behavior}

We prove Lemma \ref{lem:zeta-transducer-from-transducer} which states that the behavior of the infinite transducer is the limit of the behavior of the corresponding finite transducer, that is $\L^\ZZ(\T,A,B,Q) = \limlang{\L(\T,A,B,Q,Q,Q)}$.

\begin{proof}[Proof of Lemma \ref{lem:zeta-transducer-from-transducer}]
  Let $(\T,A,B,Q)$ be a $\ZZ$-transducer.

  Denote $X_{\T} \subset (A\times B)^\ZZ$ its behavior, that is $X_{\T}= \L^\ZZ(\T,A,B,Q)$.
  And $\L_{\T}\subset (A\times B)^*$ the behavior of the finite transducer $(\T,A,B,Q,Q,Q)$, that is $\L_{\T}=\L(\T,A,B,Q,Q,Q)$.

  We prove $X_{\T}=\limlang{\L_{\T}}$.

  This holds because an infinite run in $(\T,A,B,Q)$ is precisely the limit of an increasing sequence of finite runs in $(\T,A,B,Q,Q,Q)$.

  We prove that $\limlang{\L_{\T}} \subseteq X_{\T}$.
  First recall that since all states are initial and final, $\L_{\T}$ is pruned and factor-closed. Hence any sequence $(\shift{k_n}{}(u_n))_{n\in\NN}$ of words in $\L_{\T}$ that converge to a bi-infinite word $x$ can be taken as strictly increasing that is $\shift{k_n}{}(u_n)\internal{1}\shift{k_{n+1}}{}(u_{n+1})$.
  Take such a sequence $(\shift{k_n}{}(u_n))$ in $\L_{\T}$.
  For each $n$, $u_n=(a_{n,0},b_{n,0})\dots (a_{n,k},b_{n,k})$ admits a run in $(\T,A,B,Q,Q,Q)$ denote $u_n'$ the word in $(A\times B \times Q)$ such that for each $i$ we have $q_{n,i} \xrightarrow[a_{n,i}]{b_n{i}} q_{n,i+1}$.

  As $(A\times B \times Q)^\ZZ$ is compact, the sequence $(\shift{k_n}{}(u_n'))_{n\in\NN}$ has subsequence that converges to some bi-infinite word $y$.
  By hypothesis the projection of $y$ on the first two components of each letter is $x$ and $y$ is the trace of an infinite run in $(\T,A,B,Q)$ and hence $x \in X_{\T}$.

  \medskip

  Now we prove the converse, that is $X_{\T} \subseteq \limlang{\L_{\T}}$.
  Let $x = ((a_i,b_i))_{i\in \ZZ} \in X_{\T}$, there exists a sequence of states $(q_i)_{i\in\NN}$ such that
  $… \xrightarrow[a_{-1}]{b_{-1}}q_0 \xrightarrow[a_0]{b_0} q_1 \xrightarrow[a_1]{b_1}$ is a bi-infinite run in $(\T,A,B,Q)$.
  For each $n$, let $k_n = n$ and $u_n=(a_{-n},b_{-n})\dots (a_n,b_n)$ (the central subword of length $2n+1$ in $x$).

  For each $n$ we have $q_{-n-1} \xrightarrow[a_{-n}]{b_{-n}} q_{-n} \dots \xrightarrow[a_n]{b_n} q_n$ which is a run of $(\T,A,B,Q,Q,Q)$ of length $2n+1$.
  Hence $u_n\in \L_{\T}$. Therefore $x=\lim\limits_{n\to\infty} \shift{k_n}{}(u_n)$ is in $\limlang{\L_{\T}}$

\end{proof}

\subsection{Uniqueness of the Minimal Rooted Right-Resolving Pruned Presentation}\label{app:inf-uniqueness}

We recall that a presentation of a sofic subshift $X$ is simply a $\ZZ$-transducer with $\one$ for output alphabet, so we start by studying completeness in the case where $B = \one$. A presentation of a sofic subshift is:
\begin{description}
	\item[Pruned] if its language is pruned, meaning that for every word $w$ in the language, there exists $u$ and $v$ both non-empty such that $uwv$ is in the language.
	\item[Right-Resolving] whenever its transition relation is a partial function.
	\item[Rooted] whenever there exists a root state $r$ such that
	\begin{itemize}
		\item every accepted word admits an accepted run starting by $r$,
		\item and every state is accessible from $r$, meaning there is a run from $r$ to that state.
	\end{itemize}
\end{description}

\begin{lemma}[Pruned Presentation]\label{applem:pruned}
	For any pruned presentation $(\T,A,Q)$ of the sofic subshift $X$, the language recognized by the presentation is exactly $\fact{X}$.
\end{lemma}
\begin{proof}
	Let $L$ be the language recognized by $(\T,A,Q)$. By definition, we have $\limlang{L} = X$. 
	By \Cref{prop:bijection-subshifts}, $L = \fact{\limlang{L}} = \fact{X}$.
\end{proof}

\begin{theorem}[Uniqueness]\label{appthm:minimal-uniqueness-inf}
	Every non-empty sofic subshift admits a unique minimal rooted right-resolving pruned presentation, up to isomorphism.
\end{theorem}
\begin{proof}
	Given a sofic subshift $X$, let us consider the minimal deterministic automaton $(\D,A,Q,\{i\},F)$ that recognizes $\fact{X}$. Within that automaton, we write $\xrightarrow[\text{\small $\T$}]{}$ for the transitions, and $\xrightarrow[\text{\small $\D$}]{w}$ for the iterated transitions.	
	
	If we take $w = w_1\dots w_n \in \fact{X}$, it corresponds to a run from $i \xrightarrow[\text{\small $\D$}]{w} f \in F$. Since $\fact{X}$ is factor-closed, we also have $w_1\dots w_k \in \fact{X}$ for $0 \leq k \leq n$, and since $\D$ is deterministic it means that all the states of the run  $i \xrightarrow[\text{\small $\D$}]{w} f$ are also final states. This means that if we restrict our automaton to only its final states, so $(\D,A,F,\{i\},F)$, then the language recognized is identical. Note that $\D$ is now a partial function (except in the case $\fact{X} = A^*$).
	
	We now consider $q\in F$ and $u \in A^*$ such that there exists a $f \in F$ satisfying $p \xrightarrow[\text{\small $\D$}]{u} f$. In the minimal deterministic automaton, every state is accessible (otherwise the automaton would not be minimal), so there exists $v \in A^*$ such that $i \xrightarrow[\text{\small $\D$}]{v} q$, meaning that $vu \in \fact{X}$. Since $\fact{X}$ is factor-closed, this means that $u \in \fact{X}$, so there exists a run $i \xrightarrow[\text{\small $\D$}]{v} p$ for some $p \in F$. This means that if we extend initial states to be $F$, so the non-deterministic automaton $(\D,A,F,F,F)$, then the language recognized is identical. Additionally, any word accepted by this automaton admits an accepting run that start with $i$, and every state is accessible from $i$. Said otherwise, $(\D,A,F)$ is a representation of the sofic subshift $\limlang{\fact{X}}$, that is pruned, right-resolving and rooted in $i$. Using \Cref{prop:bijection-subshifts}, $\limlang{\fact{X}} = X$.
	
	Now for uniqueness, we consider another minimal pruned right-resolving presentation $(\M,A,P)$, rooted in $r$. If $\M$ is already total (which only happens whenever $\fact{X} = A^*$), $(\M,A,P,\{r\},P)$ is the minimal deterministic automaton recognizing $\fact{X}$. Assuming $\M$ is not a total function, we build build from it a deterministic automaton $(\M_\bot,A,P \sqcup \{\bot\},\{r\},P)$ where $\M_\bot$ is the total function that extends $\M$ by adding a transition toward the new state $\bot$ whenever $\M$ would be undefined. This deterministic automaton recognizes the same language as the presentation, which according to \Cref{applem:pruned} is exactly $\fact{X}$. We argue it is the minimal deterministic automaton recognizing $\fact{X}$. Indeed, if we take a deterministic automaton recognizing $\fact{X}$ that is strictly smaller, then by following the same procedure we used above on $\D$ to make it a pruned right-resolving presentation, we would obtain one that is strictly smaller than $\M$, which contradict minimality.
	
	In both cases, we have a minimal deterministic automaton that recognizes the same language as $(\D,A,P,\{i\},F)$, so by uniqueness of the minimal automata (see \cite{DBLP:books/daglib/0016921}) we have and isomorphism $\iota$ between the two. This directly induces an isomorphism between $(\D,A,F)$ and $(\M,A,P)$.
\end{proof}

\section{Proof for Bi-Infinite Words}
\label{app:infinite}

\subsection{From $\ZZ$-Transducers to Graphical Representation}\label{app:inf-zeta-transducer-red}

In this section, we prove \Cref{prop:zeta-transducer-red}, that is for every $\ZZ$-transducer $\T$, its behavior can be obtained using the lift of $\T$ as follows:
\[ \input{rel-trans-inf-red.tikz} \]

We write $\R$ for the left-hand-side. We take $w \in A^\ZZ$ and $v \in B^\ZZ$, and note that by definition of the composition of relations, $w~\R~v$ if and only there exists two words $q,u \in Q^\ZZ$ such that $(w,q)~\T^\ZZ~(v,u)$ and $q~\transpose{\left(\shift{}{}\right)}~u$. 

Let us focus on the latter, it is equivalent to the for all $k \in \ZZ$, $q_{k+1} = u_k$.
Let us focus on the former, it is equivalent to, for every $k \in \ZZ$, $(w_k,q_k)~\T~(v_k,u_k)$, that is $q_k \xrightarrow[v_k]{w_k} u_k$. Combining both, we obtain $q_k \xrightarrow[v_k]{w_k} q_{k+1}$, which yields exactly the definition of a run. So we have that $w~\R~v$ if and only if $w~\L^\ZZ(\T,A,B,Q)~v$. \qed

\subsection{Explicit Definition of the Graphical Language on Bi-Infinite Words}\label{app:infinite-language}

We start by defining formally the graphical language $\ZZ$\bfup{-Trans}. It is a category where the objects are finite sets, and where the morphisms are generated by composing sequentially and in parallel the generators of \Cref{appfig:graph-generators-inf} together with the equations of a strict symmetric monoidal category (see \Cref{appfig:rel-equations-inf}), equations of a feedback  category (see \Cref{appfig:graph-feedback-inf}), and equations ensuring we faithfully embed \bfup{FinRel} (see \Cref{appfig:graph-finrel-inf}).

\begin{figure*}[h]
	\input{graph-generators-inf.tikz}
	\caption{Generators of  $\ZZ$\bfup{-Trans}.}
	\label{appfig:graph-generators-inf}
\end{figure*}

\begin{figure*}[h]
	\input{rel-equations-inf.tikz}
	\caption{Equations for a Strict Symmetric Monoidal Category.}
	\label{appfig:rel-equations-inf}
\end{figure*}

\begin{figure*}[h]
	\input{graph-feedback-inf.tikz}
	\caption{Equations for a Feedback Category.}
	\label{appfig:graph-feedback-inf}
\end{figure*}

\begin{figure*}[h]
	\input{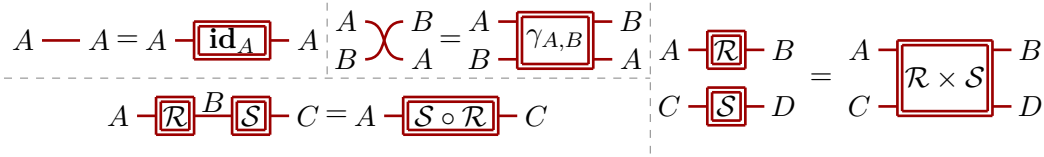}
	\caption{Equations for Faithfully Embedding \bfup{FinRel}.}
	\label{appfig:graph-finrel-inf}
\end{figure*}

Let us point a couple of facts about this language:
\begin{itemize}
	\item A bundle of wires labeled $A_1,\dots,A_n$ is the same as a single wire labeled $\Pi_{i=1}^n A_i$, and wires labeled $\one$ can be omitted.
	\item The wires and every syntactical construct are in \bfup{\red{thick red}}, to distinguish them from actual relations. In particular, $\R$ refers to an actual relation while $\rR$ refers to an element of our language.
	\item Double-line boxes denote the generator, while single-line boxes denote any diagram potentially constituted of many generators, including feedbacks.
	\item The arrows on the feedback are a reminder that those are not the same as the cup and cap of \Cref{sec:relations}.
	\item Within \Cref{appfig:graph-feedback-inf}, the gray lines correspond to bracketing, and similarly to \Cref{fig:rel-equations} the overall consequences of those equations is that bracketing can be safely ignored. 
	\item The equations of \Cref{appfig:graph-finrel-inf} ensures that if a diagram does not contain any instance of the ``feedback'' generator, we can merge all the generators into a single double-line box.
	\item The current equational theory is incomplete, additional equations will be added in \Cref{fig:inf-simulation-principle} to obtain completeness.
\end{itemize}

In order to ensure that the equations are not contradictory\footnote{Which would lead to the trivial category where all the diagrams are equal to one another.}, we provide a semantics and prove soundness of our equations. The semantics is a strong symmetric monoidal functor from $\ZZ$\bfup{-Trans} to \bfup{Rel}, which we write $\rinterp{-}$, and is actually simply ``removing the color and adding a $\_^\ZZ$ everywhere''. We provide an explicit definition in \Cref{appfig:inf-interp}. 

\begin{figure*}
	\[\input{inf-interp.tikz}\]
	\caption{Inductive Definition of the Semantics $\rinterp{-} : \ZZ\bfup{-Trans} \to \ZZ\bfup{-Rel}$.}
	\label{appfig:inf-interp}
\end{figure*}

\subsection{Soundness of the Standard Equations}\label{app:inf-soundness}

We now prove the soundness of the equational theory. 
Soundness means that if one rewrites a diagram $\rR$ into $\rS$ using any of the listed equations, we still have $\rinterp{\rR} = \rinterp{\rS}$.  The soundness of the equation of \Cref{appfig:rel-equations-inf} follows immediately from the fact that $\bfup{Rel}$ is a strict symmetric monoidal category, so we only need to look at the equations of \Cref{appfig:graph-feedback-inf}. The top-left and bottom-left equations also follow from the fact that $\bfup{Rel}$ is a strict symmetric monoidal category. The top-right is sound because $\shift{}{}_\one = \id_\one$, and the bottom-right is sound because $\shift{}{}_C \times \shift{}{}_D = \shift{}{}_{C \times D}$. \qed

\subsection{Quasi-Normal Form}\label{app:inf-quasi-normal-form}

We now explicitly state the generalization of \Cref{prop:normal-form-fin} to the infinite case.

\begin{proposition}[Quasi-Normal Form]\label{prop:normal-form-inf}
	Any diagram of $\rR \in \ZZ$\bfup{-Trans} from $A$ to $B$ can be put in the following form for some finite set $Q$ and $\T \in \bfup{FinRel}$.
	\[ \input{graph-from-trans-inf.tikz}\]
\end{proposition}
\begin{proof}
We start by using the first equation of \Cref{appfig:graph-feedback-inf} from right to left to push all the feedbacks at the bottom of the diagram. Then, we use the equations of \Cref{appfig:graph-finrel-inf} to merge all the non-feedback into a single box. Lastly, we use the last equation of \Cref{appfig:graph-feedback-inf} to merge all the feedbacks into a single feedback.
\end{proof}

\subsection{Universality}\label{app:inf-universality}

We now explicitly state generalization of \Cref{thm:universality-fin} to the infinite case.
\begin{theorem}[Universality]\label{thm:universality-inf}
	For any $\S \in \ZZ$\bfup{-Trans}, $\rinterp{\rS}$ is a sofic relation. 
	For any sofic relation $\R$, there exists $\S \in \ZZ$\bfup{-Trans} such that $\rinterp{\rS} = \R$.
\end{theorem}
\begin{proof}
It directly follows from \Cref{prop:normal-form-inf} and \Cref{prop:zeta-transducer-red}.
\end{proof}

\subsection{Soundness of the Simulation Principle}\label{app:inf-simulation-principle}

\begin{figure*}[h]
	\[\input{inf-simulation-principle.tikz}\]
	\caption{Simulation Principle for Bi-Infinite Words, and an Equation Deducible from it.} 
	\label{appfig:inf-simulation-principle}
\end{figure*}

\begin{figure*}[h]
	\[\input{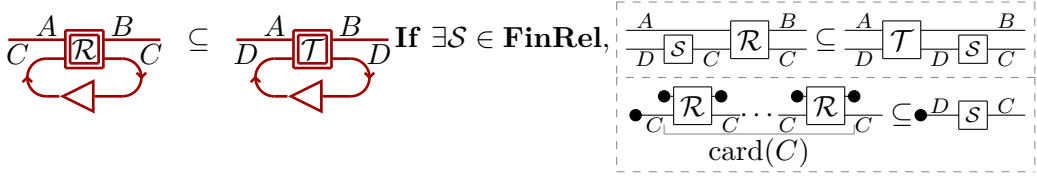}\]
	\caption{Backward-Simulation Principle for Bi-Infinite Words.}
	\label{appfig:inf-backward-simulation-principle}
\end{figure*}

\begin{figure*}[h]
	\[\input{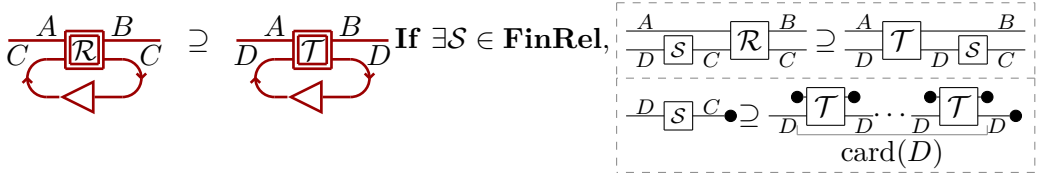}\]
	\caption{Forward-Simulation Principle for Bi-Infinite Words.}
	\label{appfig:inf-forward-simulation-principle}
\end{figure*}

We start by proving the following compactness result.
For any $\ZZ$-transducer $(\T,A,B,Q)$, we recall that by definition, we have $(a_n)_{n \in \ZZ}~\L^\ZZ(\T,A,B,Q)~(b_n)_{n \in \ZZ}$ if and only if there exists $(q_n)_{n \in \ZZ}$ such that $\forall n \in \ZZ$, $(a_n,q_{n+1})~\T~(b_n,q_n)$.
\begin{theorem}[Compactness]\label{appthm:compactness} 
	It is sufficient to only consider finite sequences, meaning that $(a_n)_{n \in \ZZ}~\L^\ZZ(\T,A,B,Q)~(b_n)_{n \in \ZZ}$  if and only if there exists for all $N \leq M \in \ZZ$ a finite sequence $(q_{n}^{N,M})_{N \leq n \leq M}$ such that $\forall N \leq n < M \in \ZZ$, $(a_n,q^{N,M}_{n+1})~\T~(b_n,q^{N,M}_n)$.
\end{theorem}
\begin{proof}
	The direct implication is trivial, as we simply take $q_{n}^{N,M} := q_n$. For the indirect implication, it is in fact enough to only consider the finite sequences centered on zero, that is $(q_{n}^{-N,N})_{-N \leq n \leq N}$ for $N \geq 0$, and ignore those that are not centered on zero. From them, we will define inductively $(q_n)_{n \in \ZZ}$. In order to do so, we add an additional induction hypothesis: there exists an infinite amount of $N \geq n$ such that $\forall -n \leq k \leq n$, $q^{-N,N}_{k} = q_k$.
	
	We start by arbitrarily providing a total order $\leq_Q$ over the finite set $Q$. 
	
	\bfup{Base Case}: We look at all the $q \in Q$ and check whether there exists an infinite amount of $N \geq 0$ such that $q^{-N,N}_0 = q$. Since $Q$ is finite, there must exists a $q$ satisfying that condition. We take $q_0$ to be the minimal $q$ (for $\leq_Q$) satisfying that condition.
	
	\bfup{Inductive Case}: We assume that we have defined $q_{-n},\dots,q_0,\dots,q_n$ such that:
	\begin{itemize}
		\item For all $-n \leq k < n$, $(a_n,q_{n+1})~\T~(b_n,q_n)$.
		\item There exists an infinite amount of $N \geq n$ such that $\forall -n \leq k \leq n$, $q^{-N,N}_{k} = q_k$.
	\end{itemize} 
	We look at all the pairs $(q_{\ominus},q_{\oplus}) \in Q^2$ and check whether there exists an infinite amount of $N \geq n+1$ such that $q^{-N,N}_{-n-1} = q_{\ominus}$ and $q^{-N,N}_{n+1} = q_{\oplus}$. Since $Q^2$ is finite, there must exists a pair $(q_{\ominus},q_{\oplus})$ satisfying that condition. We take $(q_{-n-1},q_{n+1})$ to be the minimal pair (for the lexicographic order) satisfying that condition. By definition, both induction hypothesis are preserved.
\end{proof}

We start by proving the soundness of \Cref{appfig:inf-backward-simulation-principle} with respect to $\rinterp{-}$. We need to prove that if we consider $(a_n)_{n \in \ZZ}~\L^\ZZ(\R,A,B,C)~(b_n)_{n \in \ZZ}$ then  $(a_n)_{n \in \ZZ}~\L^\ZZ(\T,A,B,D)~(b_n)_{n \in \ZZ}$. 

Focusing on the right-hand-side, and using the above \Cref{appthm:compactness}, it is enough to show there exists for all $N \leq M \in \ZZ$ a finite sequence $(d_{n}^{N,M})_{N \leq n \leq M}$ such that $\forall N \leq n < M \in \ZZ$, $(a_n,d^{N,M}_{n+1})~\T~(b_n,d^{N,M}_n)$. Or, in diagrammatic terms, it is enough to show that the following relation relates $(a_N,\dots,a_M)$ to $(b_N,\dots,b_M)$.
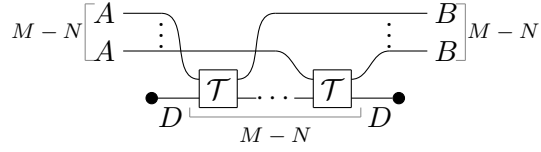
\[ \input{inf-simulation-principle-sound-right.tikz} \]

Focusing on the left-hand-side, it means that there exists $(c_n)_{n \in \ZZ}$ such that $\forall n \in \ZZ$, $(a_n,c_{n+1})~\T~(b_n,c_n)$. It implies that for $N \leq M$, the following relation relates $(a_N,\dots,a_M)$ to $(b_N,\dots,b_M)$.
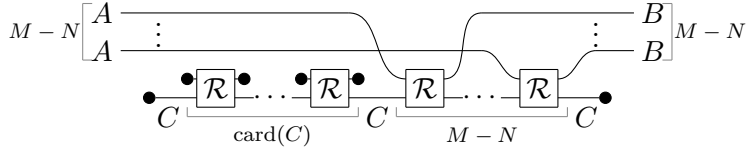
\[ \input{inf-simulation-principle-sound-left.tikz} \]
Using the bottom precondition of \Cref{appfig:inf-backward-simulation-principle}, we can replace the left part of the above diagram by $\S \circ \bullet_D$ and the result relation would still relate $(a_N,\dots,a_M)$ to $(b_N,\dots,b_M)$. Then using the top precondition we can slide the $\S$ toward the right and transform every $\R$ into a $\T$, and finally using that the full relation is maximal we can replace $\bullet_C \circ \S$ by $\bullet_D$. The resulting relation still relates $(a_N,\dots,a_M)$ to $(b_N,\dots,b_M)$. \qed

Proving the soundness of \Cref{appfig:inf-forward-simulation-principle} follows an identical proof, albeit mirrored. And the soundness of the simulation principle of \Cref{appfig:inf-simulation-principle} follows from both the backward and forward principle. \qed

\subsection{Deducing the Sliding Equation}\label{app:inf-slide}

We prove that the bottom equation from \Cref{appfig:inf-simulation-principle} can be deduced from the others. For that, we start by using \Cref{prop:normal-form-inf} on $\rR$:
\[\input{graph-slide-proof-normal-form-inf.tikz}\]
We can then merge the two feedbacks into one, rewriting both sides of the equation we are trying to prove into:
\[\input{graph-slide-proof-rewrite-inf.tikz}\]

We then conclude using the the simulation principle, using the following prerequisites and the fact that $\bullet$ is the maximal relation for the inclusion:
\[\input{graph-slide-proof-prerequisites-inf.tikz}\]
~\qed

\subsection{Completeness of the Simulation Principle}\label{app:inf-completeness}

The core idea of the proof is similar to the finite case, meaning we will determinize then minimize the presentation of our sofic subshift, except we have to start by pruning the states of our presentation that cannot be part of any bi-infinite path. Whenever we consider $(\T,A,Q)$ to be a presentation of sofic subshift, we write $\xrightarrow[\text{\small $\T$}]{}$ for the transitions within that presentation, and we write $\xrightarrow[\text{\small $\T$}]{w}$ with $w$ a finite word of size $k$ for the iterated transition $\xrightarrow[\text{\small $\T$}]{w_1} \dots \xrightarrow[\text{\small $\T$}]{w_k}$.

\begin{definition}[Forward-Pruning]
	Let $(\T,A,Q)$ be a presentation of a sofic subshift. Its forward-pruning is the presentation $(\T,A,Q^{\to \infty})$, where $Q^{\to \infty}$ is the restriction of $Q$ to the states $q_0$ that are the start of an infinite path $q_0 \xrightarrow[\text{\small $\T$}]{a_1} q_1 \xrightarrow[\text{\small $\T$}]{a_2} \dots$, or equivalently of a path of size at least $\textup{card}(Q)$ as such a path necessarily loops, and where $\T$ has been restricted to $A \times Q^{\to \infty} \to Q^{\to \infty}$.
\end{definition}

\begin{proposition}\label{prop:forward-pruning-inf}
	For all presentation of a sofic subshift $(\T,A,Q)$, we have
	\[ \input{inf-forward-pruning.tikz}\]
	where $= : Q \to Q^{\to \infty}$ is the usual ``equality'' relation.
\end{proposition}
\begin{proof}
	The first equation state that if $q_0$ is the start of an infinite path, and $q_{-1} \xrightarrow[\text{\small $\T$}]{a_{0}} q_0$, then $q_{-1}$ is also the start of an infinite path, which is true.
	
	The second equation is trivially true as the right-hand-side simplifies to $\bullet_{Q^{\to\infty}}$, which is maximal for the inclusion.
	
	The third equation state that if $q_0$ is the start of a path of size $\textup{card}(Q)$, then it is the start of an infinite path, which is true because a path of size $\textup{card}(Q)$ necessarily loops.	
\end{proof}

\begin{corollary}\label{appcor:forward-pruning-inf}
	For all presentation of a sofic subshift $(\T,A,Q)$, using \Cref{fig:inf-simulation-principle}, we have
	\[ \input{inf-forward-pruning-conclusion.tikz}\]
\end{corollary}

\begin{definition}[Backward-Pruning]
	Let $(\T,A,Q)$ be a presentation of a sofic subshift. Its backward-pruning is the presentation $(\T,A,Q^{\infty\to})$, where $Q^{\infty\to}$ is the restriction of $Q$ to the states $q_0$ that are the end of an infinite path $\dots \xrightarrow[\text{\small $\T$}]{a_{-1}} q_{-1} \xrightarrow[\text{\small $\T$}]{a_0} q_0$, or equivalently of a path of size at least $\textup{card}(Q)$ as such a path necessarily loops, and where $\T$ has been restricted to $A \times Q^{\infty\to} \to Q^{\infty\to}$.
\end{definition}

\begin{proposition}\label{prop:backward-pruning-inf}
	For all presentation of a sofic subshift $(\T,A,Q)$, we have
	\[ \input{inf-backward-pruning.tikz}\]
	where $= : Q^{\infty \to} \to Q$ is the usual ``equality'' relation.
\end{proposition}
\begin{proof}
	The first equation state that if $q_0$ is the end of an infinite path, and $q_0 \xrightarrow[\text{\small $\T$}]{a_{1}} q_1$, then $q_1$ is also the end of an infinite path, which is true.
	
	The second equation is trivially true as the left-hand-side simplifies to $\bullet_{Q^{\infty\to}}$, which is maximal for the inclusion.
	
	The third equation state that if $q_0$ is the end of a path of size $\textup{card}(Q)$, then it is the end of an infinite path, which is true because a path of size $\textup{card}(Q)$ necessarily loops.
\end{proof}

\begin{corollary}\label{appcor:backward-pruning-inf}
	For all presentation of a sofic subshift $(\T,A,Q)$, using \Cref{fig:inf-simulation-principle}, we have
	\[ \input{inf-backward-pruning-conclusion.tikz}\]
\end{corollary}

\begin{definition}[Pruning]
	Let $(\T,A,Q)$ be a presentation of a sofic subshift. Its pruning is the presentation $(\T,A,Q^{\infty\to\infty})$, where $Q^{\infty\to\infty}$ is the restriction of $Q$ to the states $q_0$ that are on a bi-infinite path $\dots \xrightarrow[\text{\small $\T$}]{a_{-1}} q_{-1} \xrightarrow[\text{\small $\T$}]{a_0} q_0 \xrightarrow[\text{\small $\T$}]{a_1} q_1 \xrightarrow[\text{\small $\T$}]{a_2} \dots$, or equivalently at the middle of a path of size at least $2 \times \textup{card}(Q)$ as such a path necessarily loops before and after $q_0$, and where $\T$ has been restricted to $A \times Q^{\infty\to} \to Q^{\infty\to}$.
\end{definition}

\begin{lemma}\label{lem:double-pruning}
	The pruning of a presentation is equal to the forward-pruning of its backward-pruning, or equivalently the backward pruning of its forward-pruning. The pruning of a representation is always a \textbf{pruned} representation.
\end{lemma}
\begin{proof}
	If a state $q_0$ is on a bi-infinite path $(q_n)_{n \in \ZZ}$, then all the states $q_n$ are on an bi-infinite path, so none of them will be removed by the forward-pruning or backward-pruning, meaning that $Q^{\infty \to \infty} \subseteq (Q^{\to \infty})^{\infty \to}$. We then take $q_0 \in (Q^{\to \infty})^{\infty \to}$, there is an infinite path $(q_n)_{n \leq 0}$ in $(Q^{\to \infty})^{\infty \to}$ that ends on $q_0$, and there is an infinite path  $(q_n)_{n \geq 0}$ in $Q^{\to \infty}$ that starts with $q_0$. So $q_0 \in Q^{\infty \to \infty}$, hence $Q^{\infty \to \infty} = (Q^{\to \infty})^{\infty \to}$. The same reasoning can be done to prove $Q^{\infty \to \infty} = (Q^{\infty\to})^{\to\infty}$.
\end{proof}

\begin{definition}[Determinization]
	Let $(\T,A,Q)$ be a pruned presentation of a \bfup{non-empty} sofic subshift. Its determinization is the \textbf{pruned right-resolving rooted} presentation $(\Pne(\T), A, \Pne^{\textup{acc}}(Q))$.
	More precisely, we start by recalling that $\Pne(Q)$ is the set of all non-empty subsets of $Q$. Since our subshift is non-empty, $Q \in \Pne(Q)$. We use $x,y$ for elements of $Q$, and $X$,$Y$ for elements of $\Pne(Q)$.  Then, we define the partial function\footnote{We consider partial functions to be a special case of relations.} $\Pne(\T): A \times \Pne(Q) \to \Pne(Q)$ as 
	\[ \Pne(\T)(a,X)=\{y\in Q ~|~ \exists x\in X,~  x \xrightarrow[\text{\small $\T$}]{a} y \} \bfup{ if }\text{it is non-empty (undefined otherwise)} \]
	We then consider the set $\Pne^{\textup{acc}}(Q)$ of subsets of $Q$ accessible by iteration of that function, starting from $Q$. We can now restrict $\Pne(\T)$ to a function $A \mapsto \Pne^{\textup{acc}}(Q) \to \Pne^{\textup{acc}}(Q)$. The root is simply $Q$, and every state is accessible from $Q$, by definition of $ \Pne^{\textup{acc}}(Q)$.
\end{definition}

\begin{proposition}\label{prop:determinization-inf}
	For all pruned presentation of a non-empty sofic subshift $(\T,A,Q)$, we have
	\[ \input{inf-determinization.tikz}\]
	where $\ni : \Pne^{\textup{acc}}(Q) \to Q$ is the usual ``contains'' relation, that is $X \ni x$ whenever $x \in X$.
\end{proposition}
\begin{proof}
	In this proof, we write $\P(\T)$ for the total function from $A \times \P(Q) \to \P(Q)$ which extends $\Pne(\T)$ as below, and note that $\Pne(\T)(a,X)$ is defined if and only if $\P(\T)(a,X) \neq \varnothing$.
	\[ \P(\T)(a,X)=\{y\in Q ~|~ \exists x\in X,~  x \xrightarrow[\text{\small $\T$}]{a} y \} \]
	
	Using the logical reasoning as in \Cref{sec:logic}, we can rewrite the first equation as the following. We are looking at $\forall a\in A,~ \forall X\in \Pne^{\textup{acc}}(Q),~\forall y\in Q,~$
	
	\[ \left(\exists x\in Q,~ (x\in X) \land (x \xrightarrow[\text{\small$\T$}]{a} y)\right) ~\Leftrightarrow~ \left(\exists Y\in \Pne^{\textup{acc}}(Q),~   (X \xrightarrow[\text{\small$\Pne(\T)$}]{a} Y) \land (y\in Y)\right)\]
	We start by reformulating the right side of the equivalence, as $\Pne(\T)$ is a partial function, we can replace the $\exists$ and obtain $(\P(\T)(a,X) \neq \varnothing) \land (y \in \Pne(a,X))$, which is equivalent to $y \in \P(\T)(a,X)$. Then, using the definition of $ \P(\T)$, we obtain that it is equivalent to $\exists x \in X, x \xrightarrow[\text{\small$\T$}]{a} y$, which is exactly the left side of the equivalence. 
	
	The second and third equations are trivially true as their left-hand-side simplifies to $\bullet_{\Pne^{\textup{acc}}(Q)}$ and $\bullet_{Q}$, which are maximal for the inclusion.
\end{proof}

\begin{corollary}\label{appcor:determinization-inf}
	For all pruned presentation of a non-empty sofic subshift $(\T,A,Q)$, using \Cref{fig:fin-simulation-principle}, we have
	\[ \input{inf-determinization-conclusion.tikz}\]
\end{corollary}

\begin{definition}[Minimization]	
	Let $(\D,A,P)$ be a \textbf{pruned right-resolving rooted} presentation of a sofic subshift, and we write $r$ for its\footnote{While it is possible for that presentation to have multiple roots, choosing a different root have no influence on this minimization, as the $L_w$ defined would be identical.} root. Its minimzation is the \textbf{right-resolving rooted} presentation $(L_\D,A,L^{\neq \varnothing}_{A^*})$.
	More precisely, we write $L_w = \{ v \in A^* \mid \exists f \in P, r \xrightarrow[\text{\small $\D$}]{wv} f\}$, and remark that whenever $L_w = L_u$, then for all $a \in A$ we also have $L_{wa} = L_{ua}$. We then define $L^{\neq \varnothing}_{A^*} = \{ L_w \neq \varnothing \mid w \in A^*\}$ note that a single element of this set might correspond to multiple distinct $w$, and in fact $L_{A^*}$ is actually smaller or equal to $P$ in cardinality. Lastly, we define the transition function as $L_\D(a,L_w) =  L_{wa}$ whenever it is non-empty, and undefined otherwise.
\end{definition}

Similarly to the uniqueness of the minimal deterministic automata, this minimization yield the unique (up to isomorphism) minimal \textbf{pruned right-resolving rooted} presentation of the given sofic subshift, see \Cref{appthm:minimal-uniqueness-inf}.

\begin{proposition}
	For all pruned right-resolving rooted presentation $(\D,A,P)$, we have
	\[ \input{inf-minimization.tikz}\]
	where $\L : P \to L^{\neq \varnothing}_{A^*}$ relates $p \in P$ to $\ell \in L^{\neq \varnothing}_{A^*}$ whenever $\ell = \{ v \mid \exists f \in P, p \xrightarrow[\text{\small$\D$}]{v} f\}$. We note that since every state is accessible, there exists a $w \in A^*$ such that the later is equal to $L_w$.
\end{proposition}
\begin{proof}
	Using the logical reasoning as in \Cref{sec:logic}, we can rewrite the first equation as the following. We are looking at $\forall a \in A, \forall p \in P, \forall \ell \in L^{\neq \varnothing}_{A^*},$
	\[ \begin{array}{c} \left(\exists q\in P,~  (p \xrightarrow[\text{\small$\D$}]{a} q)\land (\ell = \{ w \mid \exists f \in F, q \xrightarrow[\text{\small$\D$}]{w} f\}) \right) \\ \Leftrightarrow \\ \left(\exists m \in L^{\neq \varnothing}_{A^*},~  (m = \{ v \mid \exists f \in P, p \xrightarrow[\text{\small$\D$}]{v} f\}) \land (m \xrightarrow[\text{\small$L_\D$}]{a} \ell)\right) \end{array} \]
	Before studying either side of that equivalence, we note that since $p$ is accessible from the root state $r$, there exists a $u \in A^*$ such that $r \xrightarrow[\text{\small$\D$}]{u} p$, hence $L_u = \{ v \mid \exists f \in F, p \xrightarrow[\text{\small$\D$}]{v} f\}$.
	
	We start by reformulating the second part of the equivalence.  Since the formula for $m$ is given, we can remove the $\exists m$ and obtain $L_u \xrightarrow[\text{\small$L_\D$}]{a} \ell$. We can now use the definition of $L_\D$ and simplify the later in $(L_{ua} \neq \varnothing) \land (\ell = L_{ua})$.
	
	We now look at the first part of the equivalence, since $\D$ is right-resolving, $q$ is uniquely determined if it exists, so we can remove $\exists q$ and we obtain $(\D(a,p) \text{ is defined}) \land (\ell = \{ w \mid \exists f \in F, \D(a,p) \xrightarrow[\text{\small$\D$}]{w} f\})$. Since  $r \xrightarrow[\text{\small$\D$}]{u} p$, this is equivalent to $(L_{ua} \neq \varnothing) \land (\ell = \{ w \mid \exists f \in F, r \xrightarrow[\text{\small$\D$}]{uaw} f\})$, that is $(L_{ua} \neq \varnothing) \land (\ell = L_{ua})$.

	The second and third equations are trivially true as their right-hand-side simplifies to $\bullet_{L_{A^*}}$ and $\bullet_{P}$, which are maximal for the inclusion.
\end{proof}

\begin{corollary}\label{appcor:minimization-inf}
	For all deterministic finite automata $(\D,A,P,\{i\},F)$ where every state is accessible from the initial state, using \Cref{fig:inf-simulation-principle}, we have
	\[ \input{inf-minimization-conclusion.tikz}\]
\end{corollary}
We can now provide a proof to \Cref{thm:completeness-inf}.
\begin{proof}
	We start we two diagrams $\rR,\rS$ of $\ZZ$\bfup{-Trans} from $A$ to $B$ and assume $\rinterp{\rR} = \rinterp{\rS}$. We start by focusing on $\rR$. We combine it with the cup $\epsilon_B$ to bend its output into an input, and we then use \Cref{prop:normal-form-inf} to put the result in normal form:
	\[ \input{inf-compl-epsilon.tikz} \qquad = \qquad \input{inf-compl-normal-form.tikz}\]
	Then, we use \Cref{appcor:forward-pruning-inf,appcor:backward-pruning-inf} on $\T$ to prune it, and we distinguish two cases.
	
	\textbf{Case 1}: If $\rinterp{\rR}$ is the empty relation, then the result of the pruning is empty, meaning that:
	\[ \input{inf-compl-epsilon.tikz} \qquad = \qquad \input{inf-compl-empty.tikz} \]
	We can do the same for $\S$, which leads to
	\[\input{inf-compl-epsilon.tikz} \qquad = \qquad \input{inf-compl-epsilon-bis.tikz}   \]
	By combining with the cap $\eta_B$ to bend the input $B$ into an output, and using the the fact that $(\id_B \times \epsilon_B) \circ (\eta_B \times \id_B) = \id_B$, we obtain $\rR = \rS$.
	
	\textbf{Case 2}: If $\rinterp{\rR}$ is non-empty, then the result of the pruning is non-empty. We can continue and use \Cref{appcor:determinization-inf} on $\T$ followed by \Cref{appcor:minimization-inf} on $\P(\T)$ to obtain:
	\[ \input{inf-compl-epsilon.tikz} \qquad = \qquad \input{inf-compl-minimal.tikz}\]
	For convenience, we name $(\R_M,B \times A,Q_M)$ the resulting minimal representation. We do the same for $\S$ and write $(\S_N,B \times A,Q_N)$ for the resulting minimal representation. Since $\rinterp{\rR} = \rinterp{\rS}$ and by soundness of the equations, we obtain that:
	\[ \rinterp{\input{inf-compl-minimal-M.tikz}} \qquad = \qquad \rinterp{\input{inf-compl-minimal-N.tikz}} \]
	Using the definition of $\rinterp{-}$ and \Cref{prop:zeta-transducer-red}, it follows that $(\R_M,B \times A,Q_M)$ and $(\S_N,B \times A,Q_N)$ represent the same subshift, hence by uniqueness of the minimal (pruned right-resolving rooted) representation of \Cref{appthm:minimal-uniqueness-inf}, we obtain that both representations are equal up to an isomorphism $\iota : Q_M \to Q_N$, hence using the simulation principle with this $\iota$, we obtain
	\[  \input{inf-compl-minimal-M.tikz} \qquad =\qquad \input{inf-compl-minimal-N.tikz} \]
	Which leads to
	\[\input{inf-compl-epsilon.tikz} \qquad = \qquad \input{inf-compl-epsilon-bis.tikz}   \]
	By combining with the cap $\eta_B$ to bend the input $B$ into an output, and using the the fact that $(\id_B \times \epsilon_B) \circ (\eta_B \times \id_B) = \id_B$, we obtain $\rR = \rS$.
\end{proof}

\end{document}